\documentclass{lmcs}
\pdfoutput=1

\usepackage{lastpage}
\lmcsdoi{17}{2}{8}
\lmcsheading{}{\pageref{LastPage}}{}{}%
{Jun.~30,~2020}{Apr.~22,~2021}{}

\keywords{
higher inductive types,
homotopy type theory,
Coq,
bicategories}

\newtheorem{axiom}[thm]{Axiom}
\theoremstyle{definition}
\newtheorem{constrInternal}[thm]{Construction}
\theoremstyle{plain}

\usepackage{mathtools}
\usepackage{amssymb}
\usepackage[utf8]{inputenc}
\usepackage{bussproofs}
\usepackage{listings}
\usepackage{coq}
\usepackage{color}
\usepackage{xcolor}
\usepackage{stmaryrd}

\usepackage{xspace}
\usepackage{xifthen}

\usepackage[all,cmtip]{xy}

\newenvironment{construction}[2][]
{\pushQED{\qed}\begin{constrInternal}[{for Problem~\ref{#2}\ifthenelse{\isempty{#1}}{}{; #1}}]}
	{\popQED\end{constrInternal}}

\newenvironment{bprooftree}
{\leavevmode\hbox\bgroup}
{\DisplayProof\egroup}

\newcommand{\etal}{\emph{et~al.}}

\newcommand{\fat}[1]{\textbf{#1}}

\newcommand{\eqdef}{:\equiv}

\newcommand{\constfont}[1]{\ensuremath{\mathsf{#1}}}
\newcommand{\cat}[1]{\ensuremath{\constfont{#1}}\xspace}

\newcommand{\tc}{\theta}
\newcommand{\tcB}{\gamma}
\newcommand{\tcC}{\tau}

\newcommand{\lunitor}[1]{\lambda}
\newcommand{\runitor}[1]{\rho}
\newcommand{\linvunitor}[1]{\lambda^{-1}}
\newcommand{\rinvunitor}[1]{\rho^{-1}}
\newcommand{\lassoc}[3]{\alpha}
\newcommand{\rassoc}[3]{\alpha^{-1}}

\newcommand{\identitor}[1]{{#1}_i}
\newcommand{\compositor}[1]{{#1}_c}

\newcommand{\id}{\operatorname{id}}
\newcommand{\vcomp}{\bullet}
\newcommand{\whiskerl}{\vartriangleleft}
\newcommand{\whiskerr}{\vartriangleright}
\newcommand{\onecell}{\rightarrow}
\newcommand{\twocell}{\Rightarrow}

\newcommand{\B}{\cat{B}}
\newcommand{\C}{\cat{C}}
\newcommand{\D}{\cat{D}}

\newcommand{\onetypes}{1\mbox{-}\cat{Type}}
\newcommand{\grpd}{\cat{Grpd}}

\newcommand{\idtoiso}{\constfont{idtoiso}}

\newcommand{\ap}[2]{\constfont{ap} \ #1 \ #2}
\newcommand{\apd}[2]{\constfont{apd} \ #1 \ #2}

\newcommand{\disp}[1]{\overline{#1}}

\newcommand{\ff}{\disp{f}}

\renewcommand{\gg}{\disp{g}}

\newcommand{\xx}{\disp{x}}
\newcommand{\yy}{\disp{y}}
\newcommand{\zz}{\disp{z}}
\newcommand{\dtc}{\disp{\tc}}

\newcommand{\FF}{\disp{F}}

\newcommand{\GG}{\disp{G}}
\newcommand{\HH}{\disp{H}}
\newcommand{\LL}{\disp{L}}
\newcommand{\RR}{\disp{R}}
\newcommand{\etaeta}{\disp{\eta}}
\newcommand{\thetatheta}{\disp{\theta}}
\newcommand{\epseps}{\disp{\epsilon}}
\newcommand{\tautaul}{\disp{\tau_1}}
\newcommand{\tautaur}{\disp{\tau_2}}
\newcommand{\mm}{\disp{m}}

\newcommand{\dob}[2]{\ensuremath{{#1}(#2)}} 
\newcommand{\dmor}[3]{#1 \xrightarrow{#3} #2} 
\newcommand{\dtwo}[3]{#1 \xRightarrow{#3} #2} 
\newcommand{\total}[2][]{\ensuremath{\textstyle \int_{#1}{#2}}} 

\newcommand{\dproj}[1]{\pi_{#1}}

\newcommand{\pseudo}{\cat{Pseudo}}
\newcommand{\pstrans}[2]{#1 \twocell #2}
\newcommand{\modif}[2]{#1 \Rrightarrow #2}
\newcommand{\disppsfun}[3]{\dmor{#1}{#2}{#3}}
\newcommand{\disppstrans}[3]{\dtwo{#1}{#2}{#3}}
\newcommand{\dispmodif}[3]{\xymatrix{#1 \ar@3[r]^-{#3} & #2}}

\newcommand{\alg}{\cat{Alg}}

\newcommand{\spac}{\hskip 0.2em plus 0.1em}
\def\Lam #1.{\lambda\,#1.\spac}%
\def\Sum #1.{\sum_{#1}\spac}%
\def\Prod #1.{\prod_{#1}\spac}%

\newcommand{\Type}[0]{\textsc{Type}}
\newcommand{\type}[1]{\operatorname{\textsf{#1}}}
\newcommand{\constructor}[1]{\operatorname{\mathbf{#1}}}

\newcommand{\defeq}{\equiv}

\newcommand{\bool}[0]{\type{2}}
\renewcommand{\booltrue}{\constructor{true}}
\renewcommand{\boolfalse}{\constructor{false}}

\newcommand{\unit}[0]{\type{1}}
\newcommand{\unitt}{\constructor{tt}}
\newcommand{\idpath}{\constructor{refl}}

\newcommand{\functions}[2]{#1 \rightarrow #2}
\newcommand{\depprod}[2]{\prod #1 , #2}

\newcommand{\refl}[1]{\constructor{idpath}(#1)}
\newcommand{\concat}[2]{#1 \bullet #2}
\newcommand{\inverse}[1]{#1^{-1}}

\newcommand{\pathover}[4][]{#3 =_{#2}^{#1} #4}
\newcommand{\globeover}[4][]{#3 =_{#2}^{#1} #4}

\newcommand{\projl}{\pi_1}
\newcommand{\projr}{\pi_2}
\newcommand{\mappair}[2]{\langle #1 , #2 \rangle}

\newcommand{\inl}{\constructor{inl}}
\newcommand{\inr}{\constructor{inr}}

\newcommand{\funextsec}{\constfont{funext}}

\newcommand{\sign}{\Sigma}
\newcommand{\pover}[1]{\overline{#1}}

\newcommand{\DFAlg}{\constfont{DFalg}}
\newcommand{\DCell}{\constfont{DFcell}}
\newcommand{\FSub}{\constfont{FSub}}

\newcommand{\pointconstr}[1][]{\constfont{A}^{#1}}
\newcommand{\pathlabel}[1][]{\constfont{J}^{#1}_{\constfont{P}}}
\newcommand{\patharg}[1][]{\constfont{S}^{#1}}
\newcommand{\pathleft}[1][]{\constfont{l}^{#1}}
\newcommand{\pathright}[1][]{\constfont{r}^{#1}}
\newcommand{\homotlabel}[1][]{\constfont{J}^{#1}_{\constfont{H}}}
\newcommand{\homotpointarg}[1][]{\constfont{R}^{#1}}
\newcommand{\homotpathtarg}[1][]{\constfont{T}^{#1}}
\newcommand{\pathargleft}[1][]{\constfont{a}^{#1}}
\newcommand{\pathargright}[1][]{\constfont{b}^{#1}}
\newcommand{\homotpathleft}[1][]{\constfont{s}^{#1}}
\newcommand{\homotpathright}[1][]{\constfont{t}^{#1}}
\newcommand{\homotleft}[1][]{\constfont{p}^{#1}}
\newcommand{\homotright}[1][]{\constfont{q}^{#1}}

\newcommand{\prealgM}{\constfont{PreAlg}}
\newcommand{\pathalgM}{\constfont{PathAlg}}
\newcommand{\pathalgMD}{\constfont{DPathAlg}}
\newcommand{\algM}{\constfont{Alg}}
\newcommand{\prealg}[1]{\prealgM(#1)}
\newcommand{\pathalg}[1]{\pathalgM(#1)}
\newcommand{\algebra}[1]{\algM(#1)}
\newcommand{\prealgG}[1]{\prealgM_{\grpd}(#1)}
\newcommand{\pathalgG}[1]{\pathalgM_{\grpd}(#1)}
\newcommand{\pathalgGD}[1]{\pathalgMD_{\grpd}(#1)}
\newcommand{\algG}[1]{\algM_{\grpd}(#1)}

\newcommand{\AlgPoint}[1]{\constfont{c}^{#1}}
\newcommand{\AlgPath}[2]{\constfont{p}^{#1}_{#2}}
\newcommand{\AlgHomot}[2]{\constfont{h}^{#1}_{#2}}

\newcommand{\AlgMapPoint}[1]{\constfont{c}^{#1}}
\newcommand{\AlgMapPath}[2]{\constfont{p}^{#1}_{#2}}

\newcommand{\AlgCellPoint}[1]{\constfont{c}^{#1}}

\newcommand{\DispAlgPoint}[1]{\overline{\constfont{c}}^{#1}}
\newcommand{\DispAlgPath}[2]{\overline{\constfont{p}}^{#1}_{#2}}
\newcommand{\DispAlgHomot}[2]{\overline{\constfont{h}}^{#1}_{#2}}
\newcommand{\TotalAlg}[2][]{\ensuremath{\textstyle \int_{#1}{#2}}} 

\newcommand{\mor}[3]{#1(#2, #3)}
\newcommand{\idgrpd}[1]{\id(#1)}
\newcommand{\compgrpd}[2]{#1 \cdot #2}

\newcommand{\quot}[2]{{#1 / #2 }}

\newcommand{\gquotType}{\type{GQuot}}
\newcommand{\gcl}{\constructor{gcl}}
\newcommand{\gcleq}{\constructor{gcleq}}
\newcommand{\gconcat}{\constructor{gconcat}}
\renewcommand{\ge}{\constructor{ge}}
\newcommand{\gtrunc}{\constructor{gtrunc}}

\newcommand{\gclY}{\constfont{gcl}_Y}
\newcommand{\gcleqY}{\constfont{gcleq}_Y}
\newcommand{\gconcatY}{\constfont{gconcat}_Y}
\newcommand{\geY}{\constfont{ge}_Y}

\newcommand{\gind}{\constructor{gind}}

\newcommand{\pgrpd}{\constfont{PathGrpd}}
\newcommand{\gquot}{\constfont{GQuot}}
\newcommand{\prepgrpd}{\constfont{PathGrpd}_{\prealgM}}
\newcommand{\pregquot}{\constfont{GQuot}_{\prealgM}}
\newcommand{\pathpgrpd}{\constfont{PathGrpd}_{\pathalgM}}
\newcommand{\pathgquot}{\constfont{GQuot}_{\pathalgM}}
\newcommand{\algpgrpd}{\constfont{PathGrpd}_{\algM}}
\newcommand{\alggquot}{\constfont{GQuot}_{\algM}}

\newcommand{\poly}{\constfont{P}}
\newcommand{\constantP}[1]{\constructor{C}(#1)}
\newcommand{\idP}{\constfont{Id}}
\newcommand{\sumP}[2]{#1 + #2}
\newcommand{\prodP}[2]{#1 \times #2}
\newcommand{\polyAct}[2]{#1(#2)}
\newcommand{\polyDact}[2]{\overline{#1}(#2)}
\newcommand{\polyDmap}[2]{\overline{#1}(#2)}
\newcommand{\polyoplax}[1]{\constfont{oplax}(#1)}

\newcommand{\pathendpoint}[3]{\constfont{E}_{#1}(#2,#3)}
\newcommand{\idENA}{\constructor{id}} 
\newcommand{\idE}[1]{\idENA_{#1}} 
\newcommand{\comp}[2]{#1 \cdot #2} 
\newcommand{\inle}{\constructor{inl}} 
\newcommand{\inre}{\constructor{inr}} 
\newcommand{\prle}{\constructor{pr}_1} 
\newcommand{\prre}{\constructor{pr}_2} 
\newcommand{\pair}[2]{(#1 , #2)} 
\newcommand{\Ce}{\constructor{c}} 
\newcommand{\constr}{\constructor{constr}} 
\newcommand{\fmap}{\constructor{fmap}} 

\newcommand{\pathendpointFun}[1]{\semE{#1}}
\newcommand{\pathendpointAct}[2]{\pathendpointFun{#1}(#2)}
\newcommand{\pathendpointDact}[2]{\pathendpointFun{\overline{#1}}(#2)}
\newcommand{\pathendpointDnat}[2]{\pathendpointFun{\overline{#1}}(#2)}

\newcommand{\homotendpoint}[6]{\constfont{H}_{#1,#2,#3,#4}(#5, #6)}

\newcommand{\hrefl}[1]{\constructor{idpath}(#1)}
\newcommand{\hinv}[1]{#1^{-1}}
\newcommand{\hconcatsymb}[0]{@}
\newcommand{\hconcat}[2]{#1 \, \hconcatsymb \, #2}
\newcommand{\hassocN}{\boldsymbol\alpha}
\newcommand{\hassoc}[3]{\hassocN(#1, #2, #3)}
\newcommand{\hlunitN}{\boldsymbol\lambda}
\newcommand{\hlunit}[1]{\hlunitN(#1)}
\newcommand{\hrunitN}{\boldsymbol\rho}
\newcommand{\hrunit}[1]{\hrunitN(#1)}
\newcommand{\hprl}[1]{\constructor{pr}_1(#1)}
\newcommand{\hprr}[1]{\constructor{pr}_2(#1)}
\newcommand{\hpairprle}{\constructor{pairpr}_1}
\newcommand{\hpairprre}{\constructor{pairpr}_2}
\newcommand{\hpair}[2]{(#1 , #2)}

\newcommand{\harg}[0]{\constructor{p}_{\constfont{arg}}}
\newcommand{\hconstr}[2]{\constructor{path}_{#1}(#2)}
\newcommand{\hcomppair}{\constructor{comppair}}
\newcommand{\hap}[2]{\constructor{ap} \ #1 \ #2}
\newcommand{\hcompconst}[1]{\constructor{cmap}(#1)}

\newcommand{\homotendpointAct}[2]{#1(#2)}
\newcommand{\homotendpointDact}[2]{\overline{#1}(#2)}

\newcommand{\idtoH}{\constfont{idtoH}}

\newcommand{\semP}[1]{\llbracket #1 \rrbracket} 
\newcommand{\semPG}[1]{\langle #1 \rangle} 
\newcommand{\semE}[1]{\llbracket #1 \rrbracket} 
\newcommand{\semEG}[1]{\langle #1 \rangle} 
\newcommand{\semH}[1]{\llbracket #1 \rrbracket} 
\newcommand{\semHG}[1]{\langle #1 \rangle} 

\newcommand{\initob}{G_0}
\newcommand{\initmor}[2]{{#1 \sim #2}}
\newcommand{\initmorgen}[3]{{#2 \sim_{#1} #3}}
\newcommand{\initeq}[2]{{#1 \approx #2}}
\newcommand{\initeqgen}[3]{{#2 \approx_{#1} #3}}
\newcommand{\initeqprop}[2]{{#1 \approx^p #2}}

\newcommand{\circleS}{S^1}
\newcommand{\baseS}{\constructor{base}_{S^1}}
\newcommand{\SLoop}{\constructor{loop}_{S^1}}

\newcommand{\torus}{\mathcal{T}^2}
\newcommand{\base}{\constructor{base}}
\newcommand{\leftLoop}{\constructor{loop_l}}
\newcommand{\rightLoop}{\constructor{loop_r}}
\newcommand{\surface}{\constructor{surf}}

\newcommand{\ZT}{\mathbb{Z}_2}
\newcommand{\ZZ}{\constructor{Z}}
\newcommand{\ZS}{\constructor{S}}
\newcommand{\ZM}{\constructor{m}}
\newcommand{\ZC}{\constructor{c}}

\newcommand{\ST}[1]{|| #1 ||_0}
\newcommand{\SC}{\constructor{inc}}
\newcommand{\Strunc}{\constructor{trunc}}

\newcommand{\PT}[1]{|| #1 ||}
\newcommand{\PC}{\constructor{inc}}

\newcommand{\freesign}[2]{\constfont{FreeSig}_{#1}(#2)}
\newcommand{\freealg}[2]{\constfont{FreeAlg}_{#1}(#2)}
\newcommand{\freealginc}[1]{\constfont{inc}_{#1}}
\newcommand{\freealgpsfun}[1]{\constfont{F}_{#1}}
\newcommand{\underlying}{\constfont{U}}
\newcommand{\freepath}[1]{\widehat{#1}}
\newcommand{\freehomot}[1]{\widehat{#1}}
\newcommand{\freemonoid}[1]{\mathsf{FreeMon}(#1)}
\newcommand{\freemonoidinc}{\constructor{inc}}

\newcommand{\coequalizer}[2]{\mathsf{Coeqz}(#1,#2)}
\newcommand{\coequalizerbase}{\constructor{inc}}
\newcommand{\coequalizerglue}{\constructor{glue}}

\newcommand{\coequifier}[2]{\mathsf{Coeqf}(#1,#2)}
\newcommand{\coequifierbase}{\constructor{inc}}
\newcommand{\coequifierglue}{\constructor{glue}}

\newcommand{\groupquot}[1]{\mathsf{GroupQuot(#1)}}
\newcommand{\groupquotbase}{\constructor{base}}
\newcommand{\groupquotloop}{\constructor{loop}}
\newcommand{\groupquotloope}{\constructor{loope}}
\newcommand{\groupquotloopm}{\constructor{loopm}}

\newcommand{\monobj}{\mathsf{MonObj}}
\newcommand{\monobjunit}{\constructor{u}}
\newcommand{\monobjtensor}{\constructor{m}}
\newcommand{\monobjlambda}{\constructor{lam}}
\newcommand{\monobjrho}{\constructor{rho}}
\newcommand{\monobjalpha}{\constructor{al}}
\newcommand{\monobjtr}{\constructor{tr}}
\newcommand{\monobjpent}{\constructor{pent}}

\newcommand{\cohgroup}{\mathsf{CohGrp}}
\newcommand{\cohgroupunit}{\constructor{u}}
\newcommand{\cohgrouptensor}{\constructor{m}}
\newcommand{\cohgroupinv}{\constructor{i}}

\newcommand{\cohgrouplinv}{\constructor{linv}}
\newcommand{\cohgrouprinv}{\constructor{rinv}}

\newcommand{\circlegrpd}{\mathcal{S}}
\newcommand{\torusgrpd}{\mathcal{T}}
\newcommand{\grquotgrpd}{\mathcal{G}}

\newcommand{\circlegrpdbase}{\mathfrak{b}}
\newcommand{\circlegrpdloop}{\mathfrak{l}}

\newcommand{\torusgrpdbase}{\mathfrak{b}}
\newcommand{\torusgrpdloopl}{\mathfrak{l}}
\newcommand{\torusgrpdloopr}{\mathfrak{r}}
\newcommand{\torusgrpdsurf}{\mathfrak{s}}

\newcommand{\prodB}[2]{#1 \times #2}
\newcommand{\projlB}{\projl}
\newcommand{\projrB}{\projr}

\renewcommand{\Im}{\constfont{Im}}
\newcommand{\ImProj}[1]{\constfont{proj}_{#1}}
\newcommand{\ImInc}[1]{\constfont{inc}_{#1}}
\newcommand{\Ker}{\constfont{Ker}}

\newcommand{\toGrpdPathAlg}[1]{\widehat{#1}}
\newcommand{\toGrpdAlg}[1]{\widehat{#1}}

\newcommand{\List}{\mathsf{List}}

\newcommand{\remove}[1]{}

\newcommand{\initmorid}[1]{\constructor{id}_{{\sim}}(#1)}
\newcommand{\initmoridvar}[1]{\constructor{id}_{{\sim}}#1}
\newcommand{\initmorinv}[1]{\constructor{inv}_{{\sim}}(#1)}
\newcommand{\initmorcomp}[2]{\constructor{comp}_{{\sim}}(#1,#2)}
\newcommand{\initmorinl}[1]{\constructor{inl}_{{\sim}}(#1)}
\newcommand{\initmorinr}[1]{\constructor{inr}_{{\sim}}(#1)}
\newcommand{\initmorpair}[2]{\constructor{pair}_{{\sim}}(#1,#2)}
\newcommand{\initmorap}[1]{\constructor{ap}_{{\sim}}(#1)}
\newcommand{\initmorpath}[1]{\constructor{path}_{{\sim}}(#1)}

\begin{document}
	
\lstset{language=Coq}

\title{Constructing Higher Inductive Types as Groupoid Quotients}


\author[Niccol{\`o} Veltri]{Niccol{\`o} Veltri\rsuper{a}}
\address{\lsuper{a}Tallinn University of Technology, Estonia}
\email{niccolo@cs.ioc.ee}

\author[Niels van der Weide]{Niels van der Weide\rsuper{b}}
\address{\lsuper{b}Radboud University, Nijmegen, The Netherlands}
\email{nweide@cs.ru.nl}

\begin{abstract}
In this paper, we study finitary 1-truncated higher inductive types (HITs) in homotopy type theory. We start by showing that all these types can be constructed from the groupoid quotient. We define an internal notion of signatures for HITs, and for each signature, we construct a bicategory of algebras in 1-types and in groupoids. We continue by proving initial algebra semantics for our signatures. After that, we show that the groupoid quotient induces a biadjunction between the bicategories of algebras in 1-types and in groupoids. Then we construct a biinitial object in the bicategory of algebras in groupoids, which gives the desired algebra. From all this, we conclude that all finitary 1-truncated HITs can be constructed from the groupoid quotient.

We present several examples of HITs which are definable using our notion of signature. In particular, we show that each signature gives rise to a HIT corresponding to the freely generated algebraic structure over it. We also start the development of universal algebra in 1-types. We show that the bicategory of algebras has PIE limits, i.e. products, inserters and equifiers, and we prove a version of the first isomorphism theorem for 1-types. Finally, we give an alternative characterization of the foundamental groups of some HITs, exploiting our construction of HITs via the groupoid quotient. All the results are formalized over the UniMath library of univalent mathematics in Coq.
\end{abstract}

\maketitle

\section{Introduction}
The Martin-Löf identity type, also known as \emph{propositional equality}, represents provable equality in type theory \cite{martin1975intuitionistic}.
This type is defined polymorphically over all types and has a single introduction rule representing reflexivity.
The eliminator, often called the J-rule or path induction, is used to prove symmetry and transitivity.
Note that in particular, we can talk about the identity type of an already established identity type.
This can be iterated to obtain an infinite tower of types, which has the structure of an $\infty$-groupoid \cite{van2011types,lumsdaine2009weak}.

The J-rule is also the starting point of \emph{homotopy type theory} \cite{hottbook}.
In that setting, types are seen as spaces, terms are seen as points, proofs of identity of terms are seen as paths,
and proofs of identity between identities are seen as homotopies.
In mathematical terms, type theory can be interpreted in many Quillen model categories \cite{awodey2009homotopy,LumsdaineW15}, as for example simplicial sets \cite{simpset}.
In the simplicial model, not every two inhabitants of the identity type are equal,
which is also the case in the groupoid model \cite{HofmannS94,MR1686862} and the cubical sets model \cite{BezemCH13}.

If we assume enough axioms, then we can construct types for which we can prove that not every two inhabitants of the identity type are equal.
One example is the universe if one assumes the univalence axiom \cite{hottbook}.
Other examples can be obtained by using \emph{higher inductive types} (HITs).

Higher inductive types generalize inductive types by allowing constructors for paths, paths between paths, and so on.
While inductive types are specified by giving the arities of the operations \cite{dybjer1994inductive},
for higher inductive types one must also specify the arities of the paths, paths between paths, and so on.
The resulting higher inductive type is freely generated by the provided constructors.
To make this concrete, let us look at some examples \cite{hottbook}:

\begin{center}
\begin{lstlisting}[mathescape=true]
Inductive $\circleS$ :=
| $\baseS$ : $\circleS$
| $\SLoop$ : $\baseS = \baseS$
\end{lstlisting}

\begin{lstlisting}[mathescape=true]
Inductive $\torus$ :=
| $\base$ : $\torus$
| $\leftLoop, \rightLoop$ : $\base = \base$
| $\surface$ : $\concat{\leftLoop}{\rightLoop} = \concat{\rightLoop}{\leftLoop}$
\end{lstlisting}
\end{center}

The first one, $\circleS$, represents the circle.
It is generated by a point constructor $\baseS : \circleS$ and a path constructor $\SLoop : \baseS = \baseS$.
The second one, $\torus$, represents the torus.
This type is generated by a point constructor $\base$, two path constructors $\leftLoop$ and $\rightLoop$ of type $\base = \base$,
and a homotopy constructor $\surface : \concat{\leftLoop}{\rightLoop} = \concat{\rightLoop}{\leftLoop}$
where $p \vcomp q$ denotes the concatenation of $p$ and $q$.
Note that path and homotopy constructors depend on previously given constructors in the specification.
For both types, introduction, elimination, and computation rules can be given \cite{hottbook}.

In this paper, we study a schema of higher inductive types that allows defining types by
giving constructors for the points, paths, and homotopies.
All of these constructors can be recursive, but they can only have a finite number of recursive arguments.
Concretely, this means that every inhabitant can be constructed as a finitely branching tree.
Note that recursion is necessary to cover examples such as the set truncation, algebraic theories, and the integers.
Such HITs are called \emph{finitary}.
A similar scheme was studied by Dybjer and Moeneclaey
and they interpret HITs on this scheme in the groupoid model \cite{DBLP:journals/entcs/DybjerM18}.

Say that a type $X$ is \emph{1-truncated} if for all $x, y : X$, $p, q : x = y$, and $r, s : p = q$ we have $r = s$,
and a \emph{1-type} is a type which is 1-truncated.
In terms of the $\infty$-groupoid structure mentioned before, such types are 1-groupoids.
An example of a 1-type is $\circleS$ \cite{LicataS13}, which we mentioned before,
and another one is the classifying space of a group \cite{LicataF14}.
Groupoids are related to 1-types via the \emph{groupoid quotient} \cite{sojakovaPhD},
which takes a groupoid $G$ and returns a 1-type
whose points are objects of $G$ identified up to isomorphism.

The main goal of this paper is to show that finitary 1-truncated higher inductive types can be derived from simpler principles.
More specifically, every finitary 1-truncated HIT can be constructed in a type theory
with the natural numbers, propositional truncations, set quotients, and groupoid quotients.
Note that the set quotient is a special instance of the groupoid quotient.
The result of this paper can be used to simplify the semantic study of finitary 1-truncated HITs.
Instead of verifying the existence of a wide class of HITs, one only needs to check the existence
of propositional truncations and groupoid quotients.

Moreover, we employ our framework for HITs for the development of
2-dimensional universal algebra. Each HIT discussed in the paper comes
with a notion of algebra: a 1-type (or a groupoid) which is closed
under the introduction rules of the HIT.  Algebras for a HIT form a
bicategory. We prove that this bicategory has PIE limits. Moreover,
all its morphisms admit a factorization analogous to the one given by the
first isomorphism theorem. Our framework also allows the construction
of the free algebra for a signature, generalizing the notion of term
algebra in (1-dimensional) universal algebra. Notice that the bicategory of algebras (and also all the other concrete bicategories we consider in this paper) is a (2,1)-category, so the notion of PIE limit in this case coincides with that of homotopy limit \cite{AvigadKL15}.

Lastly, we show how to exploit our construction of HITs via the
groupoid quotient to calculate the fundamental group of some HITs.

The contributions of this paper are summarized as follows
\begin{itemize}
	\item An internal definition of signatures for HITs which allow path and homotopy constructors (Definition \ref{def:signature});
	\item Bicategories of algebras in both 1-types and groupoids (Definition \ref{def:bicat_grpd});
	\item A proof that biinitial algebras in 1-types satisfy the induction principle (Proposition \ref{thm:initial_alg_sem});
	\item A biadjunction between the bicategories of algebras in 1-types and algebras in groupoids (Construction \ref{constr:alg_biadj});
	\item A construction of 1-truncated HITs from the groupoid quotient (Construction \ref{constr:hit_exist}),
	which shows that such HITs exist.
	This is the main contribution of this paper;
        \item PIE limits in the bicategories of algebras in 1-types (Section \ref{sec:finite_limits});
        \item The definition of the free algebra for a signature as a HIT (Definition \ref{def:free_alg});
        \item A proof of the first isomorphism theorem for 1-types (Theorem \ref{thm:iso_thm});
        \item An alternative approach for calculating the fundamental groups of some HITs (Section~\ref{sec:fundamental_groups}).
\end{itemize}

\paragraph*{Related work.}
Various schemes of higher inductive types have been defined and studied.
Awodey \etal \ studied inductive types in homotopy type theory and proved initial
algebra semantics \cite{AwodeyGS12}.
Sojakova extended their result to various higher inductive types, among which
are the groupoid quotient, W-suspensions, and the torus \cite{Sojakova15,sojakovaPhD}.
For these types, Sojakova proved that homotopy initiality is equivalent to the induction principle while we only show that these two imply each other.
Note that all the HITs we consider, are 1-truncated, while Sojakova looked at types which are not necessarily 1-truncated.
However, while we consider a general class of HITs with higher path constructors, she only studied several examples of such HITs, namely the groupoid quotient, the torus, and higher truncations.
Basold \etal \ \cite{BasoldGW17} defined a scheme for HITs allowing for both point and path constructors,
but no higher constructors, and a similar scheme is given by
Moeneclaey \cite{moeneclaey2016schema}. 
Dybjer and Moeneclaey extended this scheme by allowing homotopy constructors and
they give semantics in the groupoid model \cite{DBLP:journals/entcs/DybjerM18}.
In the framework of computational higher-dimensional type theory \cite{AngiuliHW17},
Cavallo and Harper defined indexed cubical inductive types and proved canonicity \cite{CavalloH19}.
Altenkirch \etal \ defined quotient inductive-inductive types, which combine the features
of quotient types with inductive-inductive types \cite{forsberg2010inductive,AltenkirchCDKF18}.
Kov\'acs and Kaposi extended this syntax to higher inductive-inductive types \cite{KaposiK18},
which can be used to define not necessarily set-truncated types.
The scheme studied in this paper, is most similar to the one by Dybjer and Moeneclaey \cite{DBLP:journals/entcs/DybjerM18}
with the restriction that each type has a constructor indicating that the type is 1-truncated.
In particular, this means that inductive-inductive types are not considered.
Note that the HITs we study only have the right elimination property with respect to 1-types,
unlike W-suspensions \cite{Sojakova15,sojakovaPhD}.

Higher inductive types have already been used for numerous applications.
One of them is synthetic homotopy theory.
Spaces, such as the real projective spaces, higher spheres, and Eilenberg-MacLane spaces,
can be defined as higher inductive types \cite{licata2013pi,LicataF14,DBLP:conf/lics/BuchholtzR17,hottbook}.
The resulting definitions are strong enough to determine homotopy groups
\cite{licata2013pi,LicataS13}.
In addition, algebraic theories can be modeled as HITs, which allows one
to define Kuratowski-finite sets as a higher inductive type \cite{frumin2018finite}.
Other applications of HITs include homotopical patch theory, which provides a way
to model version control systems \cite{AngiuliMLH16}, and modeling data types
such as the integers \cite{BasoldGW17,altenkirchscoccola}.
Besides, quotient inductive-inductive types can be used to define the partiality monad \cite{AltenkirchDK17}.
These types can also be used to define type theory within type theory \cite{AltenkirchK16}
and to prove its normalization \cite{DBLP:journals/lmcs/AltenkirchK17}.
Since the HITs in this paper are 1-truncated, they are able to express term algebras of finitary algebraic theories.
For examples such as real projective spaces and higher spheres, we can only define their 1-truncation.

Several classes of higher inductive types have already been reduced to simpler ones.
Both Van Doorn and Kraus constructed propositional truncations from non-recursive higher inductive types
\cite{Doorn16,Kraus16}.
Van Doorn, Von Raumer, and Buchholtz showed that the groupoid quotient can be constructed using pushouts \cite{DoornRB17}.
By combining their result with ours, we get that all finitary 1-truncated higher inductive types can be constructed from pushouts.
Using the join construction, Rijke constructed several examples of HITs, namely $n$-truncations, the Rezk completion,
and set quotients \cite{rijke2017join}.
Awodey \etal \ gave an impredicative construction of finitary inductive types and some HITs \cite{awodey2018impredicative}. 
Constructions of more general classes of HITs have also been given.
Assuming UIP, Kaposi \etal \ constructed all finitary quotient inductive-inductive types
from a single one \cite{KaposiKA19}, and without UIP, Van der Weide and Geuvers
constructed all finitary set truncated HITs from quotients \cite{van2019construction}.
Note that these two works only concern set truncated HITs while our work concerns 1-truncated HITs.
Furthermore, the HITs considered by Van der Weide and Geuvers are a special case of
the HITs in this paper.

An alternative way to verify the existence of higher inductive types,
is by constructing them directly in a model.
Coquand \etal \ interpreted several HITs in the cubical sets model \cite{BezemCH13,CoquandHM18}.
Note that one can constructively prove univalence in the cubical sets model \cite{CohenCHM16}
and that cubical type theory satisfies homotopy canonicity \cite{DBLP:conf/rta/CoquandHS19}.
Furthermore, cubical type theory has been implemented in Agda with support for higher inductive types \cite{vezzosi2019cubical}.
Lumsdaine and Shulman give a semantical scheme for HITs and show that these can be interpreted
in sufficiently nice model categories \cite{lumsdaine2017semantics}.

Lastly, there are other approaches to 2-dimensional universal algebra.
Blackwell \etal \ study 2-dimensional universal algebra in a syntax-less fashion: they define 2-categories of algebras of a given 2-monad \cite{blackwell1989two}.
Their main result says that these 2-categories support limits and colimits.
On the contrast, our approach is based on a concrete notion of signature, which represents the syntax.
Note that by Corollary \ref{cor:pseudomonadofsig}, each signature gives rise to a pseudomonad and as such, the work by Blackwell \etal \ considers a more general version of 2-dimensional universal algebra.

\paragraph*{Formalization.}
All results in this paper are formalized in Coq \cite{Coq:manual} using UniMath \cite{UniMath}.
The formalization uses the version with \texttt{git} hash \href{https://github.com/UniMath/UniMath/tree/2dadfb61f5ef0d9805cf0eb6b80ef2beb26472d5}{2dadfb61} and can be found here:
\begin{center}
\url{https://github.com/nmvdw/GrpdHITs/tree/extended}
\end{center}
In our development, we slightly deviate from the UniMath philosophy by employing Coq's support for inductive type families.

\paragraph*{Overview.}
We start by recalling the groupoid quotient and displayed bicategories in Section \ref{sec:prelims}.
Displayed bicategories are our main tool to construct the bicategory of algebras for a signature.
In Section \ref{sec:signs}, we define signatures and show that each signature gives rise to a bicategory of algebras in both 1-types and groupoids.
The notion of a higher inductive type on a signature is given in Section \ref{sec:induction}.
There, we also prove initial algebra semantics, which says that biinitiality is a sufficient condition for being a HIT.
To construct the desired higher inductive type, we use the groupoid quotient, and in Section \ref{sec:biadj} we lift this to a biadjunction on the level of algebras.
As a consequence, constructing the initial algebra of a signature in groupoids is sufficient to construct the desired higher inductive type.
In Section \ref{sec:existence}, we construct the desired initial algebra and we conclude that each signature has an associated higher inductive type. 
Next we discuss more examples of HITs in Section \ref{sec:examples} and there we also show how to obtain monoidal objects and coherent 2-groups as algebras for certain signatures.
After that we study 2-dimensional universal algebra with our signatures.
More specifically, we construct PIE limits of algebras in Section \ref{sec:finite_limits}, the free algebra for a signature in Section \ref{sec:free_algebra}, and we prove the first isomorphism theorem in Section \ref{sec:isomorphism_theorem}.
The final topic we discuss, is the calculation of fundamental groups.
In Section \ref{sec:fundamental_groups}, we use the way we constructed higher inductive types to determine the fundamental group of the circle, the torus, and the group quotient.
Lastly, we conclude in Section \ref{sec:conclusion}.

\paragraph*{Publication History.}
This paper is an extended version of \cite{nmvdw2020} by the second author.
In Section \ref{sec:signs}, we changed Definition \ref{def:homotep} and in Section \ref{sec:induction}, we added Construction \ref{constr:total_alg}.
Figures \ref{fig:initmor} and \ref{fig:initmoreq} are also new.
Sections \ref{sec:examples} to \ref{sec:fundamental_groups} are new.

\paragraph*{Notation.}
In this paper, we work in dependent type theory and we assume the univalence axiom. In particular, this means that we also have function extensionality.
Let us recall some notation from HoTT which we use throughout this paper.
The identity path is denoted by $\refl{x}$, the concatenation of paths $p : x = y$ and $q : y = z$ is denoted by $p \vcomp q$, and the inverse of a path $p : x = y$ is denoted by $\inverse{p} : y = x$.
Given a type $X$ with points $x, y : X$ and paths $p, q : x = y$, we call a path $s : p = q$ a \emph{2-path}.
A \emph{proposition} is a type of which all inhabitants are equal.
A \emph{set} is a type $X$ such that for all $x, y : X$ the type $x = y$ is a proposition.
\remove{A \emph{1-type} is a type $X$ such that for all $x, y : X$ the type $x = y$ is a set.}
A \emph{homotopy} between $f, g : X \rightarrow Y$ consists of a path
$f(x) = g(x)$ for each $x : A$. By assuming function
extensionality, we have access to a map $\funextsec$ sending a
homotopy between functions $f$ and $g$ to a path $f = g$. Given a
type $A$, we write $\PT{A}$ for its \emph{propositional truncation}
and $\PC : A \to \PT{A}$ for the point constructor. We define the existential quantification
$\exists a : A.B(a)$ as $\PT{\Sum {a : A}. B(a)}$.

\section{Preliminaries}
\label{sec:prelims}
\subsection{Groupoid Quotient}
Let us start by formally introducing the groupoid quotient \cite{sojakovaPhD}.
The groupoid quotient is a higher dimensional version of the set quotient,
so let us briefly recall the set quotient.
Given a setoid $(X,R)$ (a set $X$ with an equivalence relation $R$ valued in propositions on $X$),
the set quotient gives a type $\quot{X}{R}$, which is $X$ with the points identified according to $R$.
Note that $\quot{X}{R}$ always is a set since equality in $\quot{X}{R}$ is described by $R$.

Instead of a setoid, the groupoid quotient takes a groupoid as input.
Recall that a groupoid is a category in which every morphism is invertible.
In particular, each groupoid has identity morphisms, denoted by $\idgrpd{x}$, and a composition operation.
The composition of $f$ and $g$ is denoted by $\compgrpd{f}{g}$.
In addition, the type of morphisms from $x$ to $y$ is required to be a set.
We write $\grpd$ for the type of groupoids.

In homotopy type theory, groupoids (and more generally categories) are
usually required to be univalent, meaning that the type of isomorphisms between objects
in a groupoid is equivalent to to the type of equalities between objects. In this paper we
do not use this requirement, so the groupoids are not
necessarily univalent. Note that the objects in a groupoid form a type without any restriction on its truncation level.

Given $G : \grpd$, the groupoid quotient gives a 1-type $\gquotType(G)$.
In this type, the points are objects of $G$
and these are identified according to the morphisms in $G$.
In addition, the groupoid structure must be preserved.
Informally, we define the groupoid quotient as the following HIT.

\begin{lstlisting}[mathescape=true]
Inductive $\gquotType$ $(G : \grpd)$ :=
| $\gcl$ : $\functions{G}{\gquotType(G)}$
| $\gcleq$ : $\depprod{(x, y : G) (f : \mor{G}{x}{y})}{\gcl(x) = \gcl(y)}$
| $\ge$ : $\depprod{(x : G)}{\gcleq(\idgrpd{x}) = \refl{\gcl(x)}}$
| $\gconcat$ : $\depprod{(x, y, z : G)(f : \mor{G}{x}{y})(g : \mor{G}{y}{z})}{\gcleq(\compgrpd{f}{g}) = \concat{\gcleq(f)}{\gcleq(g)}}$
| $\gtrunc$ : $\depprod{(x, y : \gquotType(G))(p, q: x = y) (r, s : p = q)}{r = s}$
\end{lstlisting}

To formally add this type to our theory, we need to provide
introduction, elimination, and computation rules for $\gquotType(G)$.
Formulating the elimination principle requires two preliminary notions.
These are inspired by the work of Licata and Brunerie \cite{licata2015cubical}.
The first of these gives paths in a dependent type over a path in the base.

\begin{defi}
\label{def:path_over}
Given a type $X : \Type$,
a type family $Y : \functions{X}{\Type}$,
points $x_1, x_2 : X$,
a path $p : x_1 = x_2$,
and points $\overline{x_1} : Y(x_1)$ and $\overline{x_2} : Y(x_2)$ over $x_1$ and $x_2$ respectively,
we define the type $\pathover[Y]{p}{\overline{x_1}}{\overline{x_2}}$ of \fat{paths over} $p$ from $\overline{x_1}$ to $\overline{x_2}$ by path induction on $p$
by saying that the paths over the identity path $\refl{x}$ from $\overline{x_1}$ to $\overline{x_2}$ are just paths $\overline{x_1} = \overline{x_2}$.
\end{defi}

Note that the groupoid quotient also has constructors for paths between paths.
This means that we also need a dependent version of 2-paths,
and inspired by the terminology of globular sets, we call these \emph{globes} over a given 2-path.
We define them as follows.

\begin{defi}
\label{def:globe_over}
Let $X$, $Y$, and $x_1, x_2$ be as in Definition \ref{def:path_over}.
Suppose that we have paths $p, q : x_1 = x_2$,
a 2-path $g : p = q$,
and paths $\pover{p} : \pathover{p}{\overline{x_1}}{\overline{x_2}}$ and $\pover{q} : \pathover{q}{\overline{x_1}}{\overline{x_2}}$ over $p$ and $q$ respectively.
We define the type $\globeover{g}{\pover{p}}{\pover{q}}$ of \fat{globes over} $g$ from $\pover{p}$ to $\pover{q}$ by path induction on $g$
by saying that the paths over the identity path $\refl{p}$ are just paths $\pover{p} = \pover{q}$.
\end{defi}

From this point on, we assume that our type theory has the groupoid quotient.
More specifically, we assume the following axiom.

\begin{figure*}[t]
Introduction rules: 

\vspace{2pt}

\begin{center}
\begin{bprooftree}
\AxiomC{$x : G$}
\UnaryInfC{$\gcl(x) : \gquotType(G)$}
\end{bprooftree}
\begin{bprooftree}
\AxiomC{$x, y : G$}
\AxiomC{$f : \mor{G}{x}{y}$}
\BinaryInfC{$\gcleq(f) : \gcl(x) = \gcl(y)$}
\end{bprooftree}
\begin{bprooftree}
\AxiomC{$x : G$}
\UnaryInfC{$\ge(x) : \gcleq(\idgrpd{x}) = \refl{\gcl(x)}$}
\end{bprooftree}
\end{center}

\vspace{5pt}

\begin{center}
\begin{bprooftree}
\AxiomC{$x, y, z : G$}
\AxiomC{$f : \mor{G}{x}{y}$}
\AxiomC{$g : \mor{G}{y}{z}$}
\TrinaryInfC{$\gconcat(f, g) : \gcleq(\compgrpd{f}{g}) = \concat{\gcleq(f)}{\gcleq(g)}$}
\end{bprooftree}
\end{center}

\vspace{5pt}

\begin{center}
\begin{bprooftree}
\AxiomC{$x, y : \gquotType(G)$}
\AxiomC{$p, q : x = y$}
\AxiomC{$r, s : p = q$}
\TrinaryInfC{$\gtrunc(r, s) : r = s$}
\end{bprooftree}
\end{center}

\vspace{15pt}

Elimination rule:

\vspace{2pt}

\begin{center}
\begin{bprooftree}
\AxiomC{$Y : \gquotType(G) \rightarrow \onetypes$}
\AxiomC{$\gclY : \depprod{(x : G)}{Y(\gcl(x))}$}
\noLine
\BinaryInfC{$\gcleqY : \depprod{(x, y : G) (f : \mor{G}{x}{y})}{\pathover{\gcleq(f)}{\gclY(x)}{\gclY(y)}}$}
\noLine
\UnaryInfC{$\geY : \depprod{(x : G)}{\globeover{\ge(x)}{\gcleqY(\idgrpd{x})}{\refl{\gclY(x)}}}$}
\noLine
\UnaryInfC{$\gconcatY : \depprod{(x, y, z : G) (f : \mor{G}{x}{y}) (g : \mor{G}{y}{z})}{}$}
\noLine
\UnaryInfC{$\qquad\qquad\qquad\qquad\globeover{\gconcat(f, g)}{\gcleqY(\compgrpd{f}{g})}{\concat{\gcleqY(f)}{\gcleqY(g)}}$}
\UnaryInfC{$\gind(\gclY, \gcleqY, \geY, \gconcatY) : \depprod{(x : \gquotType(G))}{Y(x)}$}
\end{bprooftree}
\end{center}

\vspace{15pt}

Computation rules:

For $\gcl$: $\gind(\gclY, \gcleqY, \geY, \gconcatY)(\gcl(x)) \defeq \gclY(x)$

For $\gcleq$: $\apd{(\gind(\gclY, \gcleqY, \geY, \gconcatY))}{(\gcleq(f))} = \gcleqY(f)$
\caption{Introduction, elimination, and computation rules for the groupoid quotient \cite{sojakovaPhD}.}
\label{fig:gquot}
\end{figure*}

\begin{axiom}
For each groupoid $G$ there is a type $\gquotType(G)$ which satisfies the rules in Figure \ref{fig:gquot}.
\end{axiom}

Since the groupoid quotient is a special instance of the Rezk completion, the type $\gquotType(G)$ can be constructed without assuming any axiom \cite{rezk_completion}.
However, this construction increases the universe level and only by assuming some resizing axiom, one can stay in the same universe.

Note that there are no computation rules for $\gconcat$, $\ge$, and $\gtrunc$,
because equations on homotopies follow automatically from the fact that $Y$ is a family of 1-types.

\subsection{Bicategory Theory}
The upcoming constructions make heavy use of notions from bicategory theory \cite{10.1007/BFb0074299,leinster:basic-bicats}
and in particular, the displayed machinery introduced by Ahrens \etal \ \cite{bicatjournal}.
Here we recall some examples of bicategories and the basics of displayed bicategories.

A bicategory consists of objects, 1-cells between objects, and 2-cells between 1-cells.
The type of 1-cells from $x$ to $y$ is denoted by $x \onecell y$ and the type of 2-cells from $f$ to $g$ is denoted by $f \twocell g$.
Note that the type $f \twocell g$ is required to be a set.
There are identity 1-cells and 2-cells denoted by $\id_1$ and $\id_2$ respectively. Composition of 1-cells $f$ and $g$ is denoted by $f \cdot g$,
and the vertical composition of 2-cells $\tc$ and $\tc'$ is denoted by $\tc \vcomp \tc'$.
Note that we use diagrammatic order for composition.
The left whiskering of a 2-cell $\tc$ with 1-cell $f$ is denoted by $f \whiskerl \tc$ and right whiskering of $\tc$ with a 1-cell $g$ is denoted by $\tc \whiskerr g$. 
Unitality and associativity of vertical composition of 2-cells hold strictly, while for 1-cells these laws hold only up to invertible 2-cells. Given 1-cell $f : A \onecell B$, there are invertible 2-cells $\lunitor{f} : \id_1(A) \cdot f \twocell f$ and $\runitor{f} : f \cdot \id_1(B) \twocell f$.
Given three composable 1-cells $f,g$ and $h$, there is an invertible 2-cell $\lassoc{f}{g}{h} : f \cdot (g \cdot h) \twocell (f \cdot g) \cdot h$.

Let us fix some notation before continuing.  Given bicategories $\B_1$
and $\B_2$, we write $\pseudo(\B_1, \B_2)$ for the type of
pseudofunctors from $\B_1$ to $\B_2$. Pseudofunctors preserve identity
and composition of 1-cells up to an invertible 2-cell. Given a
pseudofunctor $F : \pseudo(\B_1, \B_2)$ and an object $A$ of $\B_1$,
there are invertible 2-cells $\identitor{F} :\id_1(F(A)) \twocell
F(\id_1(A))$ and $\compositor{F} : F(f) \cdot F(g) \twocell F(f \cdot
g)$, with $f$ and $g$ two composable 1-cells. Preservation of identity
and composition of 2-cells is strict.
The type of pseudotransformations from $F$ to $G$ is denoted by
$\pstrans{F}{G}$.
The naturality square of a pseudotransformation $\theta : \pstrans{F}{G}$ commutes only up to invertible 2-cell. Given a 1-cell $f$ in $\B_1$, there is an invertible 2-cell $\theta_1(f) : \theta(A) \cdot G(f) \twocell F(f) \cdot \theta(B)$.
The type of modifications from $\theta$ to
$\theta'$ is denoted by $\modif{\theta}{\theta'}$
\cite{leinster:basic-bicats}.  Next we discuss \emph {biadjunctions}
\cite{gurski2012biequivalences,LACK2000179}.

\begin{defi}
Let $\B_1$ and $\B_2$ be bicategories.
A \fat{biadjunction} from $\B_1$ to $\B_2$ consists of
\begin{itemize}
	\item pseudofunctors $L : \pseudo(\B_1, \B_2)$ and $R : \pseudo(\B_2, \B_1)$;
	\item pseudotransformations $\eta : \pstrans{\id(\B_1)}{L \cdot R}$ and $\varepsilon : \pstrans{R \cdot L}{\id(\B_2)}$;
	\item invertible modifications
	\[
	\tau_1 : \modif{\rho(R)^{-1} \vcomp R \whiskerl \eta \vcomp \alpha(R, L, R) \vcomp \varepsilon \whiskerr R \vcomp \lambda(R)}{\id(R)}
	\]
	\[
	\tau_2 : \modif{\lambda(L)^{-1} \vcomp \eta \whiskerr L \vcomp \alpha(L, R, L)^{-1} \vcomp L \whiskerl \varepsilon \vcomp \rho(L)}{\id(L)}
	\]
\end{itemize}
The type of biadjunctions from $\B_1$ to $\B_2$ is denoted by $L \dashv R$
where $L : \pseudo(\B_1, \B_2)$ and $R : \pseudo(\B_2, \B_1)$.
If we have $L \dashv R$, we say that $L$ is \fat{left biadjoint} to $R$.
\end{defi}

Before presenting the definition of coherent biadjunction, let us introduce some notation.
Suppose, that we have a biadjunction $L \dashv R$ where $L : \pseudo(\B_1, \B_2)$.
For each $x : \B_1$ we get a 1-cell $\eta(x) : x \onecell R(L(x))$ and for $x : \B_2$, we get a 1-cell $\varepsilon(x) : L(R(x)) \onecell x$.
Given $x : \B_1$, we get an invertible 2-cell $\widehat{\tau_2(x)} : L(\eta(x)) \cdot \varepsilon(L(x)) \twocell \id_1(L(x))$ from $\tau_2$.
Similarly, we define an invertible 2-cell $\widehat{\tau_1(x)} : \eta(R(x)) \cdot R(\varepsilon(x)) \twocell \id_1(R(x))$ for $x : \B_2$.

\begin{defi}
Given bicategories $\B_1$ and $\B_2$, a pseudofunctor $L : \pseudo(\B_1, \B_2)$, and a biadjunction $L \dashv R$, we say that $L \dashv R$ is \fat{coherent} if the following 2-cells are equal to identity 2-cells
\begin{equation*}
\begin{split}
(\identitor{L} \vcomp L(\widehat{\tau_1(x)}^{-1}) \vcomp \compositor{L}^{-1}) \whiskerr \varepsilon(x)
\vcomp \alpha^{-1}
\vcomp L(\eta(R(x))) \whiskerl \varepsilon_1(\varepsilon(x))^{-1}
\vcomp \alpha
\vcomp \widehat{\tau_2(R(x))} \whiskerr \varepsilon(x)
\end{split}
\end{equation*}
\begin{equation*}
\begin{split}
\eta(x) \whiskerl \widehat{\tau_1(L(x))}^{-1}
\vcomp \alpha
\vcomp (\eta_1(\eta(x)))^{-1} \whiskerr R(\varepsilon(L(x)))
\vcomp \alpha^{-1}
\vcomp \eta(x) \whiskerl (\compositor{R} \vcomp R(\widehat{\tau_2(x)}) \vcomp \identitor{R}^{-1})
\end{split}
\end{equation*}
\end{defi}

The following bicategories are important for subsequent constructions:
$\onetypes$ and $\grpd$.

\begin{exa}
We have
\begin{itemize}
	\item a bicategory $\onetypes$ whose objects are 1-types, 1-cells are functions, and 2-cells are homotopies;
	\item a bicategory $\grpd$ of groupoids
	whose objects are groupoids, 1-cells are functors, and 2-cells are natural transformations.
\end{itemize}
\end{exa}

Note that $\onetypes$ and $\grpd$ are actually (2,1)-categories just like the other bicategories that we use in this paper.
However, in the formalization we do not use this fact so that we can reuse the work by Ahrens \etal \ \cite{bicatjournal}.

Next we discuss \emph{displayed bicategories}, which is our main tool to define bicategories of algebras for a signature.
Intuitively, a displayed bicategory $\D$ over $\B$ represents structure and properties to be added to $\B$.
Displayed bicategories generalize displayed categories to the bicategorical setting \cite{AhrensL19}.
Each such $\D$ gives rise to a total bicategory $\total{\D}$.
The full definition can be found in the paper by Ahrens \etal \ \cite{bicatjournal}.
Here, we only show a part.

\begin{defi}
Let $\B$ be a bicategory.
A \fat{displayed bicategory} $\D$ over $\B$ consists of
\begin{itemize}
	\item For each $x : \B$ a type $\dob{\D}{x}$ of \fat{objects over $x$};
	\item For each $f : x \onecell y$, $\xx : \dob{\D}{x}$ and $\yy : \dob{\D}{y}$,
	a type $\dmor{\xx}{\yy}{f}$ of \fat{1-cells over $f$};
	\item For each $\tc : f \twocell g$, $\ff : \dmor{\xx}{\yy}{f}$, and $\gg : \dmor{\xx}{\yy}{g}$, a \emph{set} $\dtwo{\ff}{\gg}{\tc}$ of \fat{2-cells over $\tc$}.
\end{itemize}
In addition, there are identity cells and there are composition and whiskering operations.
The composition of displayed 1-cells $f$ and $g$ is denoted by $f \cdot g$, the displayed identity 1-cell is denoted by $\id_1(x)$.
The vertical composition of 2-cells $\tc$ and $\tc'$ is denoted by $\tc \vcomp \tc'$, the left and right whiskering is denoted by $f \whiskerl \tc$ and $\tc \whiskerr f$ respectively,
and the identity 2-cell is denoted by $\id_2(f)$. 
\end{defi}

\begin{defi}
\label{def:totalbicat}
Let $\B$ be a bicategory and let $\D$ be a displayed bicategory over $\B$.
We define the \fat{total bicategory} $\total{D}$ as the bicategory whose objects are dependent pairs $(x, \xx)$ with $x$ in $\B$ and $\xx$ in $\dob{\D}{x}$.
The 1-cells and 2-cells in $\total{D}$ are defined similarly.
In addition, we define the \fat{projection} $\dproj{D} : \pseudo(\total{D}, \B)$ to be the pseudofunctor which takes the first component of each pair.
\end{defi}

Let us finish this section by defining the displayed bicategories we need in the remainder of this paper.
Examples \ref{ex:DFAlg} and \ref{ex:DCell} were first given by Ahrens \etal ~ \cite{bicatjournal}.

\begin{exa}
\label{ex:DFAlg}
Given a bicategory $\B$ and a pseudofunctor $F : \pseudo(\B, \B)$,
we define a displayed bicategory $\DFAlg(F)$ over $\B$ such that
\begin{itemize}
	\item the objects over $x : \B$ are 1-cells $h_x : F(x) \onecell x$;
	\item the 1-cells over $f : x \onecell y$ from $h_x$ to $h_y$ are invertible 2-cells $\tcC_f : h_x \cdot f \twocell F(f) \cdot h_y$;
	\item the 2-cells over $\tc : f \twocell g$ from $\tcC_f$ to $\tcC_g$ are equalities
	\begin{equation*}\label{eq:twocell_alg}
	h_x \whiskerl \tc \vcomp \tcC_g = \tcC_f \vcomp F(\tc) \whiskerr h_y.
	\end{equation*}
\end{itemize}
\end{exa}

\begin{exa}
\label{ex:disp_depprod}
Given a bicategory $\B$, a type $I$, and for each $i : I$ a displayed bicategory $\D_i$ over $\B$,
we define a displayed bicategory $\depprod{(i : I)}{\D_i}$ over $\B$ such that
\begin{itemize}
	\item the objects over $x : \B$ are functions $\depprod{(i : I)}{\D_i(x)}$;
	\item the 1-cells over $f : x \onecell y$ from $\overline{x}$ to $\overline{y}$ are functions $\depprod{(i : I)}{\dmor{\overline{x}(i)}{\overline{y}(i)}{f}}$;
	\item the 2-cells over $\tc : f \twocell g$ from $\overline{f}$ to $\overline{g}$ are functions $\depprod{(i : I)}{\dtwo{\overline{f}(i)}{\overline{g}(i)}{\tc}}$.
\end{itemize}
\end{exa}

\begin{exa}
\label{ex:DCell}
Let $\B$ be a bicategory with a displayed bicategory $\D$ over it.
Now suppose that we have pseudofunctors $S, T : \pseudo(\B, \B)$ and two pseudotransformations $l, r : \pstrans{\dproj{D} \cdot S}{\dproj{D} \cdot T}$.
Then we define a displayed bicategory $\DCell(l,r)$ over $\total{D}$ such that
\begin{itemize}
	\item the objects over $x$ are 2-cells $\tcB_x : l(x) \twocell r(x)$;
	\item the 1-cells over $f : x \onecell y$ from $\tcB_x$ to $\tcB_y$ are equalities
	\[
	(\tcB_x \whiskerr T(\dproj{D}(f))) \vcomp r(f)
	=
	l(f) \vcomp (S(\dproj{D}(f)) \whiskerl \tcB_y);
	\]
	\item the 2-cells over $\tc : f \twocell g$ are inhabitants of the unit type.
\end{itemize}
\end{exa}

\begin{exa}
\label{ex:fullsub}
Let $\B$ be a bicategory and let $P$ be a family of propositions on the objects of $\B$.
We define a displayed bicategory $\FSub(P)$ over $\B$ whose objects over $x$ are proofs of $P(x)$
and whose displayed 1-cells and 2-cells are inhabitants of the unit type.
The total bicategory $\total{\FSub(P)}$ is the \fat{full subbicategory} of $\B$ whose objects satisfy $P$.
\end{exa}

\section{Signatures and their Algebras}
\label{sec:signs}
Before we can discuss how to construct 1-truncated higher inductive types,
we need to define signatures for those.
Our notion of signature is similar to the one by Dybjer and Moeneclaey \cite{DBLP:journals/entcs/DybjerM18}.
However, instead of defining them externally, we define a type of signatures within type theory
just like what was done for inductive-recursive and inductive-inductive
definitions \cite{Dybjer1999AFA,forsberg2012finite}. 
We also show that each signature $\sign$ gives rise to a bicategory of algebras for~$\sign$.

In this section, we study HITs of the following shape
\begin{lstlisting}[mathescape=true]
Inductive $H$ :=
| $c$ : $\functions{\polyAct{P}{H}}{H}$
| $p$ : $\depprod{(j : I) (x : \polyAct{Q_i}{H})}{l_i(x) = r_i(x)}$
| $s$ : $\depprod{(j : J) (x : \polyAct{R_j}{H}) (r : a_j(x) = b_j(x))}{\homotendpointAct{p_j}{x, r} = \homotendpointAct{q_j}{x, r}}$
| $t$ : $\depprod{(x, y : H) (q_1, q_2 : x = y) (r_1, r_2 : q_1 = q_2)}{r_1 = r_2}$
\end{lstlisting}
To see what the challenges are when defining such HITs, let us take a closer look at the torus.

\begin{lstlisting}[mathescape=true]
Inductive $\torus$ :=
| $\base$ : $\torus$
| $\leftLoop, \rightLoop$ : $\base = \base$
| $\surface$ : $\concat{\leftLoop}{\rightLoop} = \concat{\rightLoop}{\leftLoop}$
\end{lstlisting}

There is a point constructor $\base$, two paths constructors $\leftLoop, \rightLoop : \base = \base$,
and a homotopy constructor $\surface : \concat{\leftLoop}{\rightLoop} = \concat{\rightLoop}{\leftLoop}$.
Note that $\leftLoop$ and $\rightLoop$ refer to $\base$ and that $\surface$ refers to all other constructors.
Hence, the signatures we define must be flexible enough to allow such dependencies.

A similar challenge comes up when defining the bicategory of algebras for a signature.
For the torus, an algebra would consist of a type $X$, a point $b$, paths $p, q : b = b$, and a 2-path $s : \concat{p}{q} = \concat{q}{p}$.
Again there are dependencies: $p$ and $q$ depend on $b$ while $s$ depends on both $p$ and $q$.
To deal with these dependencies, we use displayed bicategories, which allow us to construct the bicategory of algebras in a stratified way.

\subsection{Signatures}\label{sec:signatures}
Now let us define signatures, and to do so, we must specify data which describes the constructors for points, paths, and homotopies.
To specify the point constructors, we use \emph{polynomial codes}.
Given a type $X$ and a polynomial code $P$, we get another type $\polyAct{P}{X}$.
Such a code $P$ describes the input of an operation of the form $\polyAct{P}{X} \rightarrow X$.

\begin{defi}
The type of \fat{codes for polynomials} is inductively generated by the following constructors
\[ 
\constantP{X} : \poly, \quad \idP : \poly, \quad \sumP{P_1}{P_2} : \poly, \quad \prodP{P_1}{P_2} : \poly
\]
where $X$ is a 1-type and $P_1$ and $P_2$ are elements of $\poly$.
\end{defi}

The constructor $\constantP{X}$ represents the constant polynomial returning the type $X$, $\idP$ represents the identity,
and $\sumP{P_1}{P_2}$ and $\prodP{P_1}{P_2}$ represent the sum and product respectively.
Note that we restrict ourselves to finitary polynomials since we do not have a constructor which represents the function space.

The second part of the signature describes the path constructors, which represent universally quantified equations.
To describe them, we must give two \emph{path endpoints}.
These endpoints can refer to the point constructor, which we represent by a polynomial $A$.
In addition, they have a source (the type of the quantified variable) and a target (the type of the term).
The source and the target are represented by polynomials $S$ and $T$ respectively.

\begin{figure*}[t]
\begin{center}
\begin{bprooftree}
\AxiomC{$P : \poly$}
\UnaryInfC{$\idE{A} : \pathendpoint{A}{P}{P}$}
\end{bprooftree}
\begin{bprooftree}
\AxiomC{$P, Q, R : \poly$}
\AxiomC{$e_1 : \pathendpoint{A}{P}{Q}$}
\AxiomC{$e_2 : \pathendpoint{A}{Q}{R}$}
\TrinaryInfC{$\comp{e_1}{e_2} : \pathendpoint{A}{P}{R}$}
\end{bprooftree}
\end{center}

\vspace{5pt}

\begin{center}
\begin{bprooftree}
\AxiomC{$\vphantom{P, Q : \poly}$}
\UnaryInfC{$\constr : \pathendpoint{A}{A}{\idP}$}
\end{bprooftree}
\begin{bprooftree}
\AxiomC{$P, Q : \poly$}
\UnaryInfC{$\inle : \pathendpoint{A}{P}{\sumP{P}{Q}}$}
\end{bprooftree}
\begin{bprooftree}
\AxiomC{$P, Q : \poly$}
\UnaryInfC{$\inre : \pathendpoint{A}{Q}{\sumP{P}{Q}}$}
\end{bprooftree}
\end{center}

\vspace{5pt}

\begin{center}
\begin{bprooftree}
\AxiomC{$P, Q : \poly$}
\UnaryInfC{$\prle : \pathendpoint{A}{\prodP{P}{Q}}{P}$}
\end{bprooftree}
\begin{bprooftree}
\AxiomC{$P, Q : \poly$}
\UnaryInfC{$\prre : \pathendpoint{A}{\prodP{P}{Q}}{Q}$}
\end{bprooftree}
\end{center}

\vspace{5pt}

\begin{center}

\begin{bprooftree}
\AxiomC{$P, Q, R: \poly$}
\AxiomC{$e_1 : \pathendpoint{A}{P}{Q}$}
\AxiomC{$e_2 : \pathendpoint{A}{P}{R}$}
\TrinaryInfC{$\pair{e_1}{e_2} : \pathendpoint{A}{P}{\prodP{Q}{R}}$}
\end{bprooftree}
\end{center}

\vspace{5pt}

\begin{center}
\begin{bprooftree}
\AxiomC{$P : \poly$}
\AxiomC{$X : \onetypes$}
\AxiomC{$x : X$}
\TrinaryInfC{$\Ce(x) : \pathendpoint{A}{P}{\constantP{X}}$}
\end{bprooftree}
\begin{bprooftree}
\AxiomC{$X , Y : \onetypes$}
\AxiomC{$f : X \rightarrow Y$}
\BinaryInfC{$\fmap(f) : \pathendpoint{A}{\constantP{X}}{\constantP{Y}}$}
\end{bprooftree}
\end{center}
\caption{Rules for the path endpoints.}
\label{fig:path_ep}
\end{figure*}

\begin{defi}
Let $A$, $S$, and $T$ be codes for polynomials.
The type $\pathendpoint{A}{S}{T}$ of \fat{path endpoints} with arguments $A$, source  $S$, and target $T$ is inductively generated by the constructors given in Figure \ref{fig:path_ep}.
\end{defi}

Note that the parameter $A$ is only used in the path endpoint $\constr$, which represents the point constructor.
If we have a type $X$ with a function $c : \polyAct{A}{X} \rightarrow X$,
then each endpoint $e$ gives for every $x : \polyAct{S}{X}$ a point $\pathendpointFun{e}(x) : \polyAct{T}{X}$.
Note that $\pathendpointFun{e}(x)$ depends on $c$ while we do not write $c$ in the notation.
Often we write $e(x)$ instead of $\pathendpointFun{e}(x)$.
Hence, two endpoints $l, r : \pathendpoint{A}{S}{T}$ represent the equation
\[
\depprod{(x : \polyAct{S}{X})}{\pathendpointAct{l}{x} = \pathendpointAct{r}{x}}.
\]
Note that a HIT could have arbitrarily many path constructors and we index them by the type $J$.

The last part of the signature describes the homotopy constructors
and these depend on both the point and path constructors.
A homotopy constructor represents an equation of paths, which is universally quantified over both points and paths of
the HIT being defined.
The point argument is represented by a polynomial $R$, and the path argument is represented by a polynomial $T$ and endpoints $a, b : \pathendpoint{A}{R}{T}$.
Lastly, the type of the paths in the equation is described by two endpoints $s, t : \pathendpoint{A}{R}{W}$ with a polynomial $W$.

\begin{figure*}[t]
\begin{center}
\begin{bprooftree}
\AxiomC{$T : \poly$}
\AxiomC{$e : \pathendpoint{A}{R}{T}$}
\BinaryInfC{$\hrefl{e} : \homotendpoint{l}{r}{a}{b}{e}{e}$}
\end{bprooftree}
\begin{bprooftree}
\AxiomC{$T : \poly$}
\AxiomC{$e_1, e_2 : \pathendpoint{A}{R}{T}$}
\AxiomC{$h : \homotendpoint{l}{r}{a}{b}{e_1}{e_2}$}
\TrinaryInfC{$\hinv{h} : \homotendpoint{l}{r}{a}{b}{e_2}{e_1}$}
\end{bprooftree}
\end{center}

\vspace{5pt}

\begin{center}
\begin{bprooftree}
\AxiomC{$T : \poly$}
\AxiomC{$e_1, e_2, e_3 : \pathendpoint{A}{R}{T}$}
\AxiomC{$h_1 : \homotendpoint{l}{r}{a}{b}{e_1}{e_2}$}
\AxiomC{$h_2 : \homotendpoint{l}{r}{a}{b}{e_2}{e_3}$}
\QuaternaryInfC{$\hconcat{h_1}{h_2} : \homotendpoint{l}{r}{a}{b}{e_1}{e_3}$}
\end{bprooftree}
\end{center}

\vspace{5pt}

\begin{center}
\begin{bprooftree}
\AxiomC{$T_1, T_2 : \poly$}
\AxiomC{$e_1, e_2 : \pathendpoint{A}{Q}{T_1}$}
\AxiomC{$e : \pathendpoint{A}{T_1}{T_2}$}
\AxiomC{$h : \homotendpoint{l}{r}{a}{b}{e_1}{e_2}$}
\QuaternaryInfC{$\hap e h : \homotendpoint{l}{r}{a}{b}{\comp{e_1}{e}}{\comp{e_2}{e}}$}
\end{bprooftree}
\end{center}

\vspace{5pt}

\begin{center}
\begin{bprooftree}
\AxiomC{$T_1, T_2, T_3 : \poly$}
\AxiomC{$e_1 : \pathendpoint{A}{R}{T_1}$}
\AxiomC{$e_2 : \pathendpoint{A}{T_1}{T_2}$}
\AxiomC{$e_3 : \pathendpoint{A}{T_2}{T_3}$}
\QuaternaryInfC{$\hassoc{e_1}{e_2}{e_3} : \homotendpoint{l}{r}{a}{b}{\comp{e_1}{(\comp{e_2}{e_3})}}{\comp{(\comp{e_1}{e_2})}{e_3}}$}
\end{bprooftree}
\end{center}

\vspace{5pt}

\begin{center}
\begin{bprooftree}
\AxiomC{$T: \poly$}
\AxiomC{$e : \pathendpoint{A}{R}{T}$}
\BinaryInfC{$\hlunit{e} : \homotendpoint{l}{r}{a}{b}{\comp{\idE{R}}{e}}{e}$}
\end{bprooftree}
\begin{bprooftree}
\AxiomC{$T: \poly$}
\AxiomC{$e : \pathendpoint{A}{R}{T}$}
\BinaryInfC{$\hrunit{e} : \homotendpoint{l}{r}{a}{b}{\comp{e}{\idE{T}}}{e}$}
\end{bprooftree}
\end{center}

\vspace{5pt}

\begin{center}
\begin{bprooftree}
\AxiomC{$T_1, T_2 : \poly$}
\AxiomC{$e_1 : \pathendpoint{A}{R}{T_1}$}
\AxiomC{$e_2 : \pathendpoint{A}{R}{T_2}$}
\TrinaryInfC{$\hpairprle : \homotendpoint{l}{r}{a}{b}{\comp{\pair{e_1}{e_2}}{\prle}}{e_1}, \quad \hpairprre : \homotendpoint{l}{r}{a}{b}{\comp{\pair{e_1}{e_2}}{\prre}}{e_2}$}
\end{bprooftree}
\end{center}

\vspace{5pt}

\begin{center}
\begin{bprooftree}
\def\defaultHypSeparation{\hskip .037in}
\AxiomC{$T_1, T_2 : \poly$}
\AxiomC{$e_1, e_2 : \pathendpoint{A}{R}{T_1}$}
\AxiomC{$e_3, e_4 : \pathendpoint{A}{R}{T_2}$}
\AxiomC{$h_1 : \homotendpoint{l}{r}{a}{b}{e_1}{e_2}$}
\AxiomC{$h_2 : \homotendpoint{l}{r}{a}{b}{e_3}{e_4}$}
\QuinaryInfC{$\hpair{h_1}{h_2} : \homotendpoint{l}{r}{a}{b}{\pair{e_1}{e_3}}{\pair{e_2}{e_4}}$}
\end{bprooftree}
\end{center}

\vspace{5pt}

\begin{center}
\begin{bprooftree}
\AxiomC{$T_1, T_2 , T_3: \poly$}
\AxiomC{$e_1 : \pathendpoint{A}{R}{T_1}$}
\AxiomC{$e_2 : \pathendpoint{A}{T_1}{T_2}$}
\AxiomC{$e_3 : \pathendpoint{A}{T_1}{T_3}$}
\QuaternaryInfC{$\hcomppair : \homotendpoint{l}{r}{a}{b}{\comp{e_1}{\pair{e_2}{e_3}}}{\pair{\comp{e_1}{e_2}}{\comp{e_1}{e_3}}}$}
\end{bprooftree}
\end{center}

\vspace{5pt}

\begin{center}
\begin{bprooftree}
\AxiomC{$T: \poly$}
\AxiomC{$X: \onetypes$}
\AxiomC{$x: X$}
\AxiomC{$e : \pathendpoint{A}{R}{T}$}
\QuaternaryInfC{$\hcompconst e : \homotendpoint{l}{r}{a}{b}{\comp{e}{\Ce x}}{\Ce x}$}
\end{bprooftree}
\end{center}

\vspace{5pt}

\begin{center}
\begin{bprooftree}
\AxiomC{$j : J$}
\AxiomC{$e : \pathendpoint{A}{R}{Q_j}$}
\BinaryInfC{$\hconstr{j}{e} : \homotendpoint{l}{r}{a}{b}{\comp{e}{l(j)}}{\comp{e}{r(j)}}$}
\end{bprooftree}
\begin{bprooftree}
\AxiomC{$\vphantom{e : \pathendpoint{A}{R}{Q_j}}$}
\UnaryInfC{$\harg : \homotendpoint{l}{r}{a}{b}{a}{b}$}
\end{bprooftree}
\end{center}

\caption{Rules for the homotopy endpoints.}
\label{fig:homot_ep}
\end{figure*}

\begin{defi}
\label{def:homotep}
Suppose that we have
\begin{itemize}
	\item A polynomial $A$;
	\item A type $J$ together with for each $j : J$ a polynomial $Q_j$ and endpoints $l_j, r_j : \pathendpoint{A}{Q_j}{\idP}$;
	\item A polynomial $R$;
	\item A polynomial $T$ with endpoints $a, b : \pathendpoint{A}{R}{T}$;
	\item A polynomial $W$ with endpoints $s, t : \pathendpoint{A}{R}{W}$.
\end{itemize}
Then we define the type $\homotendpoint{l}{r}{a}{b}{s}{t}$ of \fat{homotopy endpoints} inductively by the constructors in Figure \ref{fig:homot_ep}.
\end{defi}

There are three homotopy endpoints of particular importance.
The first one is $\constructor{path}$, which represents the path constructor and it makes use of $l_j$ and $r_j$.
The second one is $\harg$, which represents the path argument and it uses $a$ and $b$.
The last one is $\constructor{ap}$ and it corresponds to the action of an endpoint on a homotopy endpoint.

The way we represent path arguments allows us to specify equations with any finite number of path arguments by only two path endpoints.
For example, two path arguments $p : x_1 = y_1$ and $q : x_2 = y_2$ is represented by one path argument of type $(x_1, x_2) = (y_1, y_2)$.

From the grammar in Figure \ref{fig:homot_ep}, we can derive
the following additional homotopy endpoints, which will be employed in forthcoming
examples. 

\vspace{5pt}

\begin{center}
\begin{bprooftree}
\AxiomC{$T_1, T_2 : \poly$}
\AxiomC{$e_1, e_2 : \pathendpoint{A}{R}{T_1}$}
\AxiomC{$e_3, e_4 : \pathendpoint{A}{R}{T_2}$}
\AxiomC{$h : \homotendpoint{l}{r}{a}{b}{\pair{e_1}{e_3}}{\pair{e_2}{e_4}}$}
\QuaternaryInfC{$\hprl{h} \eqdef \hconcat{\hinv{\hpairprle}}{\hconcat{\hap {\prle} h}{\hpairprle}} : \homotendpoint{l}{r}{a}{b}{e_1}{e_2}$}
\alwaysNoLine
\UnaryInfC{$\hprr{h} \eqdef \hconcat{\hinv{\hpairprre}}{\hconcat{\hap {\prre} h}{\hpairprre}} : \homotendpoint{l}{r}{a}{b}{e_3}{e_4}$}
\end{bprooftree}
\end{center}

\vspace{5pt}
Moreover, we have a function $\idtoH$ sending a path between endpoints $e_1 = e_2$ into an homotopy endpoint $\homotendpoint{l}{r}{a}{b}{e_1}{e_2}$, readily definable by path induction.

Given a type $X$ with a function $c : \polyAct{A}{X} \rightarrow X$ and for each $x : Q_j(X)$ a path $l_j(x) = r_j(x)$,
a homotopy endpoint $p : \homotendpoint{l}{r}{a}{b}{s}{t}$ gives rise for each point $x : \polyAct{R}{X}$ and path $w : a(x) = b(x)$
to another path $\homotendpointAct{p}{x, w} : s(x) = t(x)$.
Hence, two homotopy endpoints $p, q : \homotendpoint{l}{r}{a}{b}{s}{t}$ represent the equation
\[
\depprod{(x : \polyAct{R}{X}) (w : a(x) = b(x))}{\homotendpointAct{p}{x, w} = \homotendpointAct{q}{x, w}}
\]
Now let us put this all together and define what signatures for higher inductive types are.

\begin{defi}
\label{def:signature}
A \fat{HIT-signature} $\sign$ consists of
\begin{itemize}
	\item A polynomial $\pointconstr[\sign]$;
	\item A type $\pathlabel[\sign]$ together with for each $j : \pathlabel[\sign]$ a polynomial $\patharg[\sign]_j$ and endpoints $\pathleft[\sign]_j, \pathright[\sign]_j : \pathendpoint{\pointconstr[\sign]}{\patharg[\sign]_j}{\idP}$;
	\item A type $\homotlabel[\sign]$ together with for each $j : \homotlabel[\sign]$ polynomials $\homotpointarg[\sign]_j$ and $\homotpathtarg[\sign]_j$,
	endpoints $\pathargleft[\sign]_j, \pathargright[\sign]_j : \pathendpoint{\pointconstr[\sign]}{\homotpointarg[\sign]_j}{\homotpathtarg[\sign]_j}$
	and $\homotpathleft[\sign]_j, \homotpathright[\sign]_j : \pathendpoint{\pointconstr[\sign]}{\homotpointarg[\sign]_j}{\idP}$,
	and homotopy endpoints $\homotleft[\sign]_j, \homotright[\sign]_j : \homotendpoint{\pathleft[\sign]}{\pathright[\sign]}{\pathargleft[\sign]_j}{\pathargright[\sign]_j}{\homotpathleft[\sign]_j}{\homotpathright[\sign]_j}$.
\end{itemize}
\end{defi}

If $\sign$ is clear from the context, we do not write the superscript.
In the remainder, we show how to interpret the following HIT given a signature $\sign$:

\begin{lstlisting}[mathescape=true]
Inductive $H$ :=
| $c$ : $\functions{\polyAct{\pointconstr}{H}}{H}$
| $p$ : $\depprod{(j : \pathlabel) (x : \polyAct{\patharg_j}{H})}{\pathendpointAct{\pathleft_j}{x} = \pathendpointAct{\pathright_j}{x}}$
| $s$ : $\depprod{(j : \homotlabel) (x : \polyAct{\homotpointarg_j}{H}) (r : \pathargleft_j(x) = \pathargright_j(x))}{\homotendpointAct{\homotleft_j}{x, r} = \homotendpointAct{\homotright_j}{x, r}}$
| $t$ : $\depprod{(x, y : H) (q_1, q_2 : x = y) (r_1, r_2 : q_1 = q_2)}{r_1 = r_2}$
\end{lstlisting}

Next we consider three examples of HITs we can express with these signatures.

\begin{exa}
\label{ex:torus}
The torus is described by the signature $\torus$.
\begin{itemize}
	\item Take $\pointconstr[\torus] \eqdef \constantP{\unit}$;
	\item Take $\pathlabel[\torus] \eqdef \bool$ and for both inhabitants we take $\patharg[\torus] \eqdef \constantP{\unit}$ and $\pathleft[\torus] \eqdef \pathright[\torus] \eqdef \constr$;
	\item Take $\homotlabel[\torus] \eqdef \unit$.
	Since there are no arguments for this path constructor, we take $\homotpointarg[\torus] \eqdef \homotpathtarg[\torus] \eqdef \constantP{\unit}$ and $\pathargleft[\torus] \eqdef \pathargright[\torus] \eqdef \Ce(\unitt)$.
	Now for the left-hand side and right-hand side of this equation, we take $\hconcat{\hconstr{\booltrue}{\idENA}}{\hconstr{\boolfalse}{\idENA}}$ and $\hconcat{\hconstr{\boolfalse}{\idENA}}{\hconstr{\booltrue}{\idENA}}$ respectively. 
\end{itemize}
Notice that the usual presentation of the torus does not include the
explicit 1-truncation constructor $t$, since $\torus$ is already provably
1-truncated without its presence. A similar consideration applies to
the circle $\circleS$.
\end{exa}

\begin{exa}
\label{ex:mod}
We represent the integers modulo 2 as the following HIT:
\begin{lstlisting}[mathescape=true]
Inductive $\ZT$ :=
| $\ZZ$ : $\ZT$
| $\ZS$ : $\ZT \rightarrow \ZT$
| $\ZM$ : $\depprod{(x : \ZT)}{\ZS(\ZS(x)) = x}$
| $\ZC$ : $\depprod{(x : \ZT)}{\ZM(\ZS(x)) = \ap{\ZS}{(\ZM(x))}}$
\end{lstlisting}
Note that all constructors except $\ZZ$ are recursive.
We define a signature $\ZT$.
\begin{itemize}
	\item Take $\pointconstr[\ZT] \eqdef \sumP{\constantP{\unit}}{\idP}$;
	\item Take $\pathlabel[\ZT] \eqdef \unit$ and for its unique inhabitant we take $\patharg[\ZT] \eqdef \idP$ and
	\[
	\pathleft[\ZT] \eqdef \comp{(\comp{\inre}{\constr})}{(\comp{\inre}{\constr})}, \quad
	\pathright[\ZT] \eqdef \idENA;
	\]
	\item Take $\homotlabel[\ZT] \eqdef \unit$.
	Furthermore, we take $\homotpointarg[\ZT] \eqdef \idP$ and $\pathargleft[\ZT] \eqdef \pathargright[\ZT] \eqdef \Ce(\unitt)$.
	The endpoints $\homotpathleft$ and $\homotpathright$ encode $\ZS(\ZS(\ZS(x)))$ and $\ZS(x)$ respectively,
	and for the left-hand side and right-hand side of this equation, we take
	\[
	\hap{\constr}{(\hconcat{\hinv{\hlunitN}}{\hconcat{\hassocN}{\hconcat{\hap{\inre}{(\hconstr{\unit}{\idENA})}}{\hconcat{\hinv{\hassocN}}{\hconcat{\hlunitN}{\hlunitN}}}}})}
	\]
	\[
	\hconcat{\hinv{\hassocN}}{\hconcat{\hinv{\hassocN}}{\hconcat{\hconstr{\unitt}{\comp{\inre}{\constr}}}{\hrunitN}}}.
	\]
	respectively.
	Note that we use $\hassocN$, $\hlunitN$, and $\hrunitN$ to make the equations type check.
	If we would interpret the left-hand side and right-hand side of  the homotopy constructor in 1-types,
	then all occurrences of $\hassocN$, $\hlunitN$, and $\hrunitN$ become the identity path.
	We thus get the right homotopy constructor.
\end{itemize}
\end{exa}

\begin{exa}
\label{ex:settrunc}
Given a 1-type $A$,  the set truncation of $A$ is defined by the following HIT:
\begin{lstlisting}[mathescape=true]
Inductive $\ST{A}$ :=
| $\SC$ : $A \rightarrow \ST{A}$
| $\Strunc$ : $\depprod{(x, y : \ST{A}) (p, q : x = y)}{p = q}$
\end{lstlisting}
Note that this higher inductive type has a parameter $A$,
so the signature we define depends on a 1-type $A$ as well.
To encode the path arguments of $\Strunc$, we use the fact that giving two paths $p, q : x = y$ is the same as giving a path $r : (x, x) = (y, y)$.
Define a signature $\ST{A}$ such that
\begin{itemize}
	\item $\pointconstr[\ST{A}] \eqdef \constantP{A}$;
	\item $\pathlabel[\ST{A}]$ is the empty type;
	\item $\homotlabel[\ST{A}] \eqdef \unit$.
	In addition, there are two point arguments $\homotpointarg[\ST{A}] \eqdef \prodP{\idP}{\idP}$
	and a path argument with left-hand side $\pair{\prle}{\prle}$ and right-hand side $\pair{\prre}{\prre}$.
	For the left-hand side and right-hand side of the homotopy, we take $\hprl{\harg}$ and $\hprr{\harg}$ respectively.
\end{itemize}
\end{exa}

More examples of HIT signatures are discussed in Section \ref{sec:examples}.

\subsection{Algebras in 1-types and groupoids}
With the signatures in place, our next goal is to study the introduction rules of HITs and for that, we define bicategories of algebras for a signature.
Since we ultimately want to construct HITs via the groupoid quotient, we look at both algebras in 1-types and groupoids.

In both cases, we use a stratified approach with displayed bicategories.
Let us illustrate this by briefly describing the construction for 1-types.
On $\onetypes$, we define a displayed bicategory and we denote its total bicategory by $\prealg{\sign}$.
The objects of $\prealg{\sign}$ consist of a 1-type $X$ together with an operation $\polyAct{\pointconstr[\sign]}{X} \rightarrow X$.
Concretely, the objects satisfy the point introduction rules specified by $\sign$.
On top of $\prealg{\sign}$, we define another displayed bicategory whose total bicategory is denoted by $\pathalg{\sign}$.
Objects of $\pathalg{\sign}$ satisfy the introduction rules for both the points and the paths.
Lastly, we take a full subbicategory of $\pathalg{\sign}$ obtaining another bicategory $\alg(\sign)$ whose objects satisfy the introduction rules for the points, paths and homotopies.

To define $\prealg{\sign}$, we use Example \ref{ex:DFAlg}.

\begin{prob}
\label{prob:sem_poly}
Given $P : \poly$, to construct pseudofunctors
\[
\semP{P} : \pseudo(\onetypes, \onetypes), 
\semPG{P} : \pseudo(\grpd, \grpd).
\]
\end{prob}

\begin{construction}{prob:sem_poly}
We only discuss the case for 1-types since the case for groupoids is similar.
Given a polynomial $P$ and a type $X$, we get a type $\polyAct{P}{X}$ by induction.
The verification that this gives rise to a pseudofunctor can be found in the formalization.
\end{construction}

\begin{defi}
\label{def:prealg}
Let $\sign$ be a signature.
Then we define the bicategories $\prealg{\sign}$ and $\prealgG{\sign}$
to be the total bicategories of $\DFAlg(\semP{\pointconstr[\sign]})$
and $\DFAlg(\semPG{\pointconstr[\sign]})$ respectively.
Objects of these bicategories are called \fat{prealgebras} for $\sign$.
\end{defi}

Note that prealgebras only have structure witnessing the introduction rule for the points.
Next we look at the introduction rule for the paths.
In this case, the desired structure is added via Example \ref{ex:DCell} and to apply this construction,
we interpret path endpoints as pseudotransformations.

\begin{prob}
\label{prob:sem_endpoint}
Given $e : \pathendpoint{A}{P}{Q}$, to construct pseudotransformations
\[
\semE{e} : \pstrans{\dproj{\DFAlg(\semP{A})} \cdot \semP{P}}{\dproj{\DFAlg(\semP{A})} \cdot \semP{Q}},
\]
\[
\semEG{e} : \pstrans{\dproj{\DFAlg(\semPG{A})} \cdot \semPG{P}}{\dproj{\DFAlg(\semPG{A})} \cdot \semPG{Q}}.
\]
\end{prob}

\begin{construction}{prob:sem_endpoint}
We only discuss $\semE{e}$ since $\semEG{e}$ is defined similarly.
Given a 1-type $X$ and $c : \polyAct{A}{X} \rightarrow X$, we define the function
$\semE{e} : \polyAct{P}{X} \rightarrow \polyAct{Q}{X}$ by induction.
The verification that this gives rise to a pseudotransformation can be found in the formalization.
\end{construction}

\begin{defi}
\label{def:pathalg}
Let $\sign$ be a signature.
We use Examples \ref{ex:disp_depprod} and \ref{ex:DCell} to define displayed bicategories over $\prealg{\sign}$ and $\prealgG{\sign}$.
\[
\pathalgMD(\sign) \eqdef \depprod{(i : \pathlabel[\sign])}{\DCell(\semE{\pathleft[\sign](i)},\semE{\pathright[\sign](i)})}
\]
\[
\pathalgGD{\sign} \eqdef \depprod{(i : \pathlabel[\sign])}{\DCell(\semEG{\pathleft[\sign](i)},\semEG{\pathright[\sign](i)})}
\]
We define $\pathalgM(\sign)$ and $\pathalgG{\sign}$ to be the total bicategories of $\pathalgMD(\sign)$ and $\pathalgGD{\sign}$ respectively.
Objects of $\pathalgM(\sign)$ and $\pathalgG{\sign}$ are called \fat{path algebras} for $\sign$.
\end{defi}

\begin{prob}
\label{prob:sem_homendpoint}
Suppose that we have a homotopy endpoint $h : \homotendpoint{l}{r}{a}{b}{s}{t}$.
Given a 1-type $X$ with $c : \polyAct{A}{X} \rightarrow X$ and $p : \depprod{(j : J) (x : \polyAct{Q_j}{X})}{l_j(x) = r_j(x)}$,
to construct for each $x : \polyAct{Q}{X}$ and $w : \semE{a}(x) = \semE{b}(x)$
an equality $\semH{h}(x, w) : \semE{s}(x) = \semE{t}(x)$.

In addition, given a groupoid $G$ together with a functor $c : \semPG{A}{(G)} \rightarrow G$ and for each $j : J$ a natural transformation $\semEG{l_j}(G) \twocell \semEG{r_j}(G)$,
to construct for each object $x : \semPG{Q}{(G)}$ and morphism $w : \semEG{a}(G)(x) \rightarrow \semEG{b}(G)(x)$
a morphism $\semHG{h}(x, w) : \semEG{s}(G)(x) \rightarrow \semEG{t}(G)(x)$.
\end{prob}

\begin{construction}{prob:sem_homendpoint}
By induction.
\end{construction}

\begin{defi}
\label{def:bicat_grpd}
Let $\sign$ be a HIT signature.
We define $\algM(\sign)$ to be the full subbicategory of $\pathalgM(\sign)$
in which every object $X$ satisfies
\begin{equation*}
\begin{split}
\prod
(j : \homotlabel[\sign]) (x : \polyAct{\homotpointarg[\sign]_j}{X}) (w : \semE{\pathargleft[\sign]_j}(x) = \semE{\pathargright[\sign]_j}(x)),
\ \semH{\homotleft[\sign]_j}(x, w) = \semH{\homotright[\sign]_j}(x, w)
\end{split}
\end{equation*}
In addition, we define $\algG{\sign}$ to be the full subbicategory of $\pathalgG{\sign} $
in which every object $X$ satisfies
\begin{equation*}
\begin{split}
\prod
(j : \homotlabel[\sign]) (x : \semPG{\homotpointarg[\sign]_j}{X}) (w : \semEG{\pathargleft[\sign]_j}(x) \rightarrow \semEG{\pathargright[\sign]_j}(x)),\
\semHG{\homotleft[\sign]_j}(x, w) = \semHG{\homotright[\sign]_j}(x, w)
\end{split}
\end{equation*}
Objects of $\algM(\sign)$ and $\algG{\sign}$ are called \fat{algebras} for $\sign$.
\end{defi}

The bicategory $\algM(\sign)$ is constructed by repeatedly using Definition \ref{def:totalbicat}.
By unpacking the definition, we see that an algebra $X : \algM(\sign)$ consists of
\begin{itemize}
	\item A 1-type $X$;
	\item A function $\AlgPoint{X} : \polyAct{\pointconstr}{X} \rightarrow X$;
	\item For each $j : \pathlabel$ and point $x : \polyAct{\patharg_j}{X}$ a path
	$
	\AlgPath{X}{j}(x) : \semE{\pathleft_j}(x) = \semE{\pathright_j}(x);
	$
	\item For each $j : \homotlabel$, $x : \polyAct{\homotpointarg}{X}$ and $w : \semE{a_1}(x) = \semE{a_2}(x)$, a homotopy
	$
	\AlgHomot{X}{j} : \semH{\homotleft}(x, w) = \semH{\homotright}(x, w)
	$
\end{itemize}
Furthermore, given two algebras  $X, Y : \algM(\sign)$, an algebra morphism $f : X \onecell Y$ consists of
\begin{itemize}
	\item a map $f : X \rightarrow Y$;
	\item for each $x : \polyAct{\pointconstr[\sign]}{X}$ a path
	\[
	\AlgMapPoint{f}(x) : f(\AlgPoint{X}(x)) = \AlgPoint{Y}(\semP{\pointconstr[\sign]}(f)(x));
	\]
	\item for each $j : \pathlabel[\sign]$ and $x : \polyAct{\patharg[\sign]}{X}$ a 2-path
	\[
	\AlgMapPath{f}{j}(x) :
	\ap{f}{(\AlgPath{X}{j}(x))} \vcomp \semE{\pathright[\sign]_j}(\AlgMapPoint{f})(x)
	=
	\semE{\pathleft[\sign]_j}(\AlgMapPoint{f})(x) \vcomp \AlgPath{Y}{j}(\semP{\patharg[\sign]}(f)(x)).
	\]
\end{itemize}
Lastly, given two algebras $X, Y : \algM(\sign)$ and an algebra morphisms $f, g : X \onecell Y$, a 2-cell $\theta : f \twocell g$ in $\algM(\sign)$ consists of
\begin{itemize}
	\item for each $x : X$ a path $\theta(x) : f(x) = g(x)$;
	\item for each $x : \polyAct{\pointconstr[\sign]}{X}$ a path
	\[
	\AlgCellPoint{\theta}(x) : 
	\theta(x) \vcomp \AlgMapPoint{g}(x)
	=
	\AlgMapPoint{f}(x)
	\vcomp
	\ap{\AlgPoint{Y}}{(\semP{\pointconstr[\sign]}(\theta)(x))}
	\]
\end{itemize}

\section{Induction and Biinitiality}
\label{sec:induction}
The algebra structure only represents the introduction rule
and the next step is to define the elimination and computation rules for higher inductive types.
Furthermore, we show that biinitial algebras satisfy the induction principle.

Before we can formulate these principles, we need to define dependent actions of polynomials, path endpoints, and homotopy endpoints.
All of these constructions are done by induction and details can be found in the literature \cite{DBLP:journals/entcs/DybjerM18,hermida1998structural,van2019construction}.

\begin{prob}
\label{prob:poly_dact}
Given a type $X$, a type family $Y$ on $X$, and a polynomial $P$,
to construct a type family $\polyDact{P}{Y}$ on $\polyAct{P}{X}$.
\end{prob}

\begin{prob}
\label{prob:poly_dmap}
Given a type $X$, a type family $Y$ on $X$, a polynomial $P$, and a map $f : \depprod{(x : X)}{Y(x)}$,
to construct a map
$
\polyDmap{P}{f} : \depprod{(x : \polyAct{P}{X})}{\polyDact{P}{Y}(x)}.
$
\end{prob}

\begin{prob}
\label{prob:poly_pr}
Given a type $X$, a type family $Y$ on $X$, and a polynomial $P$, to construct a map\footnote{This operation is called $\constfont{oplax}$ since, when the type family $Y$ does not depend on $X$, it corresponds to oplax monoidality of the pseudofunctor associated to the polynomial $P$ wrt. the cartesian product $\times$.}
$
\polyoplax P : \polyAct{P}{\Sum (x : X). Y(x)} \rightarrow \Sum (x : \polyAct{P}{X}). \polyDact{P}{Y}(x)
$
such that for all $z :  \polyAct{P}{\Sum (x : X). Y(x)}$, we have $\projl(\polyoplax{P}(z)) = \polyAct{P}{\projl}(z)$.
\end{prob}
The map $\polyoplax P$ actually defines an equivalence, but we do not employ its inverse in our development.
\begin{prob}
\label{prob:endpoint_dact}
Given a type $X$, a type family $Y$ on $X$, an endpoint $e : \pathendpoint{A}{P}{Q}$, and a map $c : \polyAct{A}{X} \rightarrow X$,
to construct for each $x : \polyAct{P}{X}$ and $y : \polyDact{P}{Y}(x)$ an inhabitant $\pathendpointDact{e}{y} : \polyDact{Q}{Y}(\pathendpointAct{e}{x})$.
\end{prob}

\begin{prob}
\label{prob:endpoint_dact_natural}
Suppose, that we have polynomials $A, P, Q$, a type $X$ with a map $c_X : \polyAct{A}{X} \rightarrow X$,
and a type family $Y$ on $X$ with a map $c_Y : \depprod{(x : X)}{\polyDact{A}{Y}(x) \rightarrow Y(c_X(x))}$
and a map $f : \depprod{(x : X)}{Y(x)}$.
Given an endpoint $e : \pathendpoint{A}{P}{Q}$,
to construct an equality
\[
\pathendpointDnat{e}{f} : \polyDmap{Q}{f}(\semE{e}(x)) = \pathendpointDact{e}{\polyDmap{P}{f}(x)}.
\]
\end{prob}

\begin{prob}
\label{prob:homot_dact}
Let $\sign$ be a signature.
Let $X$ be a type with a function $c_X : \polyAct{\pointconstr[\sign]}{X} \rightarrow X$
and for each $j : \pathlabel[\sign]$ and $x : \polyAct{\patharg[\sign]_j}{X}$ a path $p_X(j, x) : \semE{\pathleft[\sign]_j}(x) = \semE{\pathright[\sign]_j}(x)$.
In addition, suppose that $Y$ is a type family on $X$,
that we have a function $c_Y : \depprod{(x : \polyAct{\pointconstr[\sign]}{X})}{\polyDact{\pointconstr[\sign]}{Y}(x) \rightarrow Y(c_X(x))}$,
and that for all $j  : \pathlabel[\sign]$ and points $x : \polyAct{\patharg[\sign]_j}{X}$ and $\xx : \polyDact{\patharg[\sign]_j}{Y}(x)$,
we have a path $p_Y : \pathover{p_X(j, x)}{\pathendpointDact{\pathleft[\sign]_j}{\xx}}{\pathendpointDact{\pathright[\sign]_j}{\xx}}$.
Furthermore, let $j : \homotlabel[\sign]$, let $x : \homotpointarg[\sign]_j(X)$ and $\xx : \polyDact{\homotpointarg[\sign]_j}{Y}(x)$ be points
and let $w : \semE{\pathargleft[\sign]_j}(x) = \semE{\pathargright[\sign]_j}(x)$
and $\disp{w} : \pathover{p}{\pathendpointDact{\pathargleft[\sign]_j}{\xx}}{\pathendpointDact{\pathargright[\sign]_j}{\xx}}$
be paths.
Then for each homotopy endpoint $h : \homotendpoint{\pathleft[\sign]}{\pathright[\sign]}{\pathargleft[\sign]_j}{\pathargright[\sign]_j}{\homotpathleft[\sign]_j}{\homotpathright[\sign]_j}$, 
to construct a path
\[
\homotendpointDact{h}{\xx, \disp{w}} : \pathover{\homotendpointAct{h}{x, w}}{\pathendpointDact{\homotpathleft[\sign]_j}{\xx}}{\pathendpointDact{\homotpathright[\sign]_j}{\xx}}.
\]
\end{prob}

With these notions in place, we define \emph{displayed algebras}.
A displayed algebra represents the input of the elimination rule.
Furthermore, we show that each displayed algebra gives rise to a total algebra and a projection.

\begin{defi}\label{def:disp_alg}
Given a signature $\sign$ and an algebra $X$ for $\sign$,
a \fat{displayed algebra} $Y$ over $X$ consists of
\begin{itemize}
	\item A family $Y$ of 1-types over $X$;
	\item For each $x : \polyAct{\pointconstr}{X}$ a map $\DispAlgPoint{Y} : \polyDact{\pointconstr}{Y}(x) \rightarrow Y(\AlgPoint{X}(x))$;
	\item For each $j : \pathlabel$, $x : \polyAct{\patharg_j}{X}$, and $\pover{x} : \polyDact{\patharg_j}{Y}(x)$, a path
	$
	\DispAlgPath{Y}{j} : \pathover{\AlgPath{X}{j}{x}}{\pathendpointDact{\pathleft_j}{\pover{x}}}{\pathendpointDact{\pathright_j}{\pover{x}}}
	$;
	\item For each $j : \homotlabel$, points $x : \polyAct{\homotpointarg_j}{X}$ and $\pover{x} : \polyDact{\homotpointarg_j}{Y}(x)$,
	and paths $w : \pathargleft_j(x) = \pathargright_j(x)$ and $\pover{w} : \pathover{w}{\pathendpointDact{\pathargleft_j}{\pover{x}}}{\pathendpointDact{\pathargright_j}{\pover{x}}}$,
	a globe
	$
	\DispAlgHomot{Y}{j} : \globeover{\AlgHomot{X}{j}{(x, w)}}{\homotendpointDact{\homotleft_j}{\xx}}{\homotendpointDact{\homotright_j}{\xx}}
	$
	over $\AlgHomot{X}{j}{(x, w)}$.
\end{itemize}
\end{defi}

\begin{rem}
The type family of a displayed algebra is required to be 1-truncated.
This means that the HITs we construct, can only be eliminated into
1-types, and as a consequence, these HITs only have the right
elimination principle with respect to 1-types.
\end{rem}

\begin{prob}
\label{prob:total_alg}
Given an algebra $X$ for a signature $\sign$ and a displayed algebra $Y$ over $X$,
to construct an algebra $\TotalAlg{Y}$ for $\sign$ and a morphism of algebras $\pi_1 : \TotalAlg{Y} \onecell X$.
\end{prob}

\begin{construction}{prob:total_alg}
\label{constr:total_alg}
We only discuss the carrier and the point constructor of $\TotalAlg{Y}$.
The carrier of $\TotalAlg{Y}$ is defined to be $\sum (x : X), Y(x)$
and the function $\AlgPoint{\TotalAlg Y}$ acts on elements $z : \polyAct{\pointconstr}{\sum (x : X), Y(x)}$ as follows:
\[
\AlgPoint{\TotalAlg Y}(z) \eqdef
(\AlgPoint{X}(\projl(\polyoplax{\pointconstr}(z)))
,
\DispAlgPoint{Y}(\projr(\polyoplax{\pointconstr}(z)))).
\]
The underlying map of the morphism $\pi_1 : \TotalAlg{Y} \onecell X$ takes the first projection of a pair.
\end{construction}

We call $\TotalAlg{Y}$ the \emph{total algebra} of $Y$ and the morphism $\pi_1$ is called the \emph{first projection}.
The output of the elimination rule and the computation rules are given by a \emph{section} to be defined in Definition \ref{def:section} below.
One might expect that, just like for the groupoid quotient, the computation rules for the paths
are given as globes over some 2-path in the base (Definition~\ref{def:globe_over}).
However, this is not the case.

This is because there is a slight discrepancy between the rules for the groupoid quotient and the HITs we discuss,
namely for the former the computation rules for the points are definitional equalities
while for the latter, these rules only hold propositionally.
This affects how we need to formulate the computation rules for the paths.

Let us illustrate this via the torus (Example \ref{ex:torus}).
The input for the elimination rule consists, among others, of a type family $Y$, a point $b : Y(\base)$,
and a path $p_l : \pathover{\leftLoop}{b}{b}$.
The elimination rule gives a map $f : \depprod{(x : \torus)}{Y(x)}$.
By the point computation rule, we have a propositional equality between $f(\base)$ and $b$.
Now the computation rule for $\leftLoop$ ought to equate $\apd{f}{\leftLoop}$ and $p_l$.
However, such an equation does not type check since $\apd{f}{\leftLoop}$ has type $\pathover{\leftLoop}{f(b)}{f(b)}$ while $p_l$ has type $\pathover{\leftLoop}{b}{b}$.
In conclusion, we cannot formulate the computation rules the same way as we did for the groupoid quotient.

Our solution to this problem is to define a type of \emph{squares} over a given 2-path similarly to Definition \ref{def:globe_over} \cite{licata2015cubical}.

\begin{defi}
Let $X$ be a type and let $Y$ be a type family on $X$.
Suppose that we are given points $x_1, x_2 : X$ and $\overline{x_1}, \overline{x_1}' : Y(x_1)$ and $\overline{x_2}, \overline{x_2}' : Y(x_2)$, paths $p, q : x_1 = x_2$ together with paths $\pover{p} : \pathover{p}{\overline{x_1}}{\overline{x_2}}$ and $\pover{q} : \pathover{q}{\overline{x_1}'}{\overline{x_2}'}$ over $p$ and $q$ respectively.
If we also have two paths $h_1 : \overline{x_1} = \overline{x_1}'$ and $h_2 : \overline{x_2} = \overline{x_2}'$ and a 2-path $g : p = q$,
then we define the type of \fat{squares} over $g$ from $\pover{p}$ to $\pover{q}$ with sides $h_1$ and $h_2$ by path induction on $g$.
\end{defi}

\begin{defi}
\label{def:section}
Let $X$ be an algebra for a given signature $\sign$ and let $Y$ be a displayed algebra over $X$.
Then a \fat{section} of $Y$ consists of
\begin{itemize}
	\item A map $f : \depprod{(x : X)}{Y(x)}$;
	\item For all $x : \polyAct{\pointconstr}{X}$, an equality $f(\AlgPoint{X}(x)) = \DispAlgPoint{Y}(\polyDmap{\pointconstr}{f}(x))$;
	\item For all $j : \pathlabel$ and $x : \polyAct{\patharg}{X}$, a square from
	$\apd{f}{(\AlgPath{X}{j}(x))}$
	to
	$\DispAlgPath{Y}{j}(\polyDmap{\patharg}{f}(x))$
	with sides
	$\pathendpointDnat{\pathleft_j}{f}(x)$
	and
	$\pathendpointDnat{\pathright_j}{f}(x)$.
\end{itemize}
\end{defi}

\begin{defi}
Let $\sign$ be a signature.
Then a \fat{1-truncated higher inductive type} for $\sign$ consists of an algebra $X$ be an algebra for $\sign$ and a proof that each displayed algebra $Y$ over $X$ has a section.
\end{defi}

Often we just say that $X$ is a HIT for $\sign$ instead of saying that $X$ is 1-truncated HIT.
With this in place, we can check whether our rules for higher inductive types
agree for the usual examples with the rules given in the literature \cite{hottbook}.
We illustrate this with the torus (Example \ref{ex:torus}) and the set truncation (Example \ref{ex:settrunc}).
In the next example, we write $p \vcomp q$ for the concatenation of dependent paths.

\begin{exa}[Example \ref{ex:torus} cont'd]
Recall the signature $\torus$ for the torus.
Let $X$ be a HIT for $\torus$.
Since $X$ is an algebra, we have a point $\base : X$, two paths $\leftLoop, \rightLoop : \base = \base$,
and a 2-path $\surface : \concat{\leftLoop}{\rightLoop} = \concat{\rightLoop}{\leftLoop}$.
This corresponds precisely to the usual introduction rules of the torus.

A family $Y$ of 1-types on $X$ together with a point $b : Y(\base)$,
two paths $l : \pathover{\leftLoop}{b}{b}$ and $r : \pathover{\rightLoop}{b}{b}$
and a globe $h : \globeover{\surface}{l \vcomp r}{r \vcomp l}$
over $\surface$ gives rise to a displayed algebra over $X$.
This corresponds to the usual input of the elimination rule of the torus.
If we have a section $s$ of $Y$, then in particular, we get a map $f_s : \depprod{(x : X)}{Y(x)}$.
We also get a path $p_s : f(\base) = b$, a square from $\apd{f}{\leftLoop}$ to $l$ and one from $\apd{f}{\rightLoop}$ to $r$.
Both squares have sides $p_s$ and $p_s$.
These are the computation rules for the points and paths of the torus.
Note that since we are looking at $1$-truncated HITs, this only gives the 1-truncation of the torus.
\end{exa}

\begin{exa}[Example \ref{ex:settrunc} cont'd]
Let $A$ be a 1-type and recall the signature $\ST{A}$.
Now let $X$ be a HIT on $\ST{A}$.
Note that an algebra for $\ST{A}$ consists of a type $Z$ together with a map $A \rightarrow Z$
and a proof that $Z$ is a set.
This means in particular, that we have a map $\SC : A \rightarrow X$ and a proof $\Strunc$
that $X$ is a set.

A family $Y$ of sets on $X$ together with a map $i : \depprod{(a  : A)}{Y(\SC(A))}$
give rise to a displayed algebra over $X$.
A section $s$ of that displayed algebra consists of a map $f_s : \depprod{(x : X)}{Y(x)}$
such that $f_s(\SC(a)) = i(a)$ for all $a : A$.
This corresponds to the usual elimination and computation rules for the
set truncation.
\end{exa}

To verify that an algebra satisfies the elimination rule, we use \emph{initial algebra semantics} \cite{hermida1998structural}.
However, this technique is usually applied in a categorical setting and it uses initial objects in categories.
Since we work in a bicategorical setting, we need to use the corresponding notion in bicategory theory: \emph{biinitiality}.

\begin{defi}
Let $\B$ be a bicategory and let $x$ be an object in $\B$.
Then we say $x$ is \fat{biinitial} if
\begin{itemize}
	\item For each object $y$ there is a 1-cell $x \onecell y$;
	\item Given 1-cells $f, g : x \onecell y$, there is a 2-cell $f \twocell g$;
	\item Given 2-cells $\tc, \tc' : f \twocell g$, there is an equality $\tc = \tc'$.
\end{itemize}
\end{defi}

Briefly, an object $x$ is biinitial if for each $y$ there is a 1-cell from $x$ to $y$, which is unique up to a unique 2-cell.
One can show that this definition is equivalent to the category $x \onecell y$ being contractible for all $y$.
Now we can formulate initial algebra semantics for our signatures.

\begin{prop}
\label{thm:initial_alg_sem}
Let $\sign$ be a signature and let $X$ be an algebra for $\sign$.
Then $X$ is a 1-truncated HIT for $\sign$ if and only if $X$ is biinitial in $\algM(\sign)$.
\end{prop}

One consequence of initial algebra semantics, is that HITs are unique up to path equality if the univalence axiom holds.
This result is a consequence of the fact that the bicategory of algebras is \emph{univalent}.
Recall that a bicategory is univalent if
equality between objects $X$ and $Y$ is equivalent to adjoint equivalences between $X$ and $Y$
and equality of 1-cells $f$ and $g$ is equivalent to invertible 2-cells between $f$ and $g$
\cite{bicatjournal}.
Using the methods employed by Ahrens \etal \ \cite{bicatjournal} one can show that the bicategory of algebras
is univalent.
Since biinitial objects are unique up to adjoint equivalence, one can conclude that HITs are unique up
to path equality.

\begin{prop}
Let $\sign$ be a signature and let $H_1$ and $H_2$ be HITs for $\sign$.
Denote the underlying algebras of $H_1$ and $H_2$ by $X_1$ and $X_2$.
Then $X_1 = X_2$.
\end{prop}

\section{Lifting the Groupoid Quotient}
\label{sec:biadj}
To construct higher inductive types, we use Proposition \ref{thm:initial_alg_sem}, which says that biinitial objects satisfy the induction principle.
We use the groupoid quotient to acquire the desired algebra.
More specifically, we construct a pseudofunctor from $\algG{\sign}$ to $\algM(\sign)$,
which is the groupoid quotient on the carrier.
We do that in such a way that the obtained pseudofunctor preserves biinitiality,
so that we obtain the HIT by constructing a biinitial object in $\algG{\sign}$.

One class of pseudofunctors which preserve biinitial objects is given by left biadjoints.
More precisely, suppose that we have bicategories $\B$ and $\C$, a left biadjoint pseudofunctor $L : \pseudo(\B, \C)$, and an object $x : \B$.
Then the object $L(x)$ is biinitial if $x$ is.

Instead of directly lifting the groupoid quotient to the level of algebras,
we first show that the groupoid quotient specifies a left biadjoint pseudofunctor and then we lift that biadjunction to the level of algebras.
To do so, we use the fact we defined the bicategory of algebras via displayed bicategories.
This way we can define the biadjunction on each part of the structure separately.

More specifically, we define the notion of \emph{displayed biadjunction} between two displayed bicategories over a biadjunction in the base,
and we show that each displayed biadjunction gives rise to a total biadjunction between the total bicategories.
Defining displayed biadjunctions requires defining displayed analogues of pseudofunctors, pseudotransformations, and invertible modification,
which were defined by Ahrens \etal \ \cite{bicatjournal}.
For this, we make use of \emph{displayed invertible 2-cells} \cite{bicatjournal}.

\begin{defi}
Let $\D_1$ and $\D_2$ be displayed bicategories over $\B_1$ and $\B_2$ respectively and let $F : \pseudo(\B_1, \B_2)$ be a pseudofunctor.
Then a \fat{displayed pseudofunctor} $\FF$ from $\D_1$ to $\D_2$ over $F$ consist of
\begin{itemize}
	\item For each $x : \B_1$ a map $\FF_0 : \D_1(x) \rightarrow \D_2(F(x))$;
	\item For all 1-cells $f : x \onecell y$, objects $\xx : \D_1(x)$ and $\yy : \D_1(y)$,
	and displayed 1-cells $\ff : \dmor{\xx}{\yy}{f}$, a displayed 1-cell $\FF_1(\ff) :  \dmor{\FF_0(\xx)}{\FF_0(\yy)}{F(f)}$;
	\item For all 2-cells $\tc : f \twocell g$, displayed 1-cells $\ff : \dmor{\xx}{\yy}{f}$ and $\gg : \dmor{\xx}{\yy}{g}$,
	and displayed 2-cells $\dtc : \dtwo{\ff}{\gg}{\tc}$, a displayed 2-cell $\FF_2(\dtc) : \dtwo{\FF_1(\ff)}{\FF_2(\gg)}{F(\tc)}$;
	\item For each $x : \B$ and $\xx : \D(x)$, a displayed invertible 2-cell
	$\identitor{\FF}(\xx) : \dmor{\id_1(\FF_0(\xx))}{\FF_1(\id_1(\xx))}{\identitor{F}(x)}$;
	\item For all $\ff : \dmor{\xx}{\yy}{f}$ and $\gg : \dmor{\yy}{\zz}{g}$. a displayed invertible 2-cell
	$\compositor{\FF}(\ff, \gg) : \dmor{\FF_1(\ff) \cdot \FF_1(\gg)}{\FF_1(\ff \cdot \gg)}{\compositor{F}(f, g)}$.
\end{itemize}
Here $\identitor{F}$ and $\compositor{F}$ denote the identitor and compositor of $F$.
In addition, several coherencies, which can be found in the formalization, need to be hold.
We denote the type of displayed pseudofunctors from $\D_1$ to $\D_2$ over $F$ by $\disppsfun{\D_1}{\D_2}{F}$.
\end{defi}

\begin{defi}
Let $\D_1$ and $\D_2$ be displayed bicategories over $\B_1$ and $\B_2$ respectively.
Suppose that we have displayed pseudofunctors $\FF : \disppsfun{\D_1}{\D_2}{F}$ and $\GG : \disppsfun{\D_1}{\D_2}{G}$
and a pseudotransformation $\theta : \pstrans{F}{G}$.
Then a \fat{displayed pseudotransformation} $\thetatheta$ from $\FF$ to $\GG$ over $\theta$ consists of
\begin{itemize}
	\item For all objects $x : \B$ and $\xx : \D_1(x)$ a displayed 1-cell $\thetatheta_0(x) : \dmor{\FF_0(\xx)}{\GG_0(\xx)}{\theta(x)}$;
	\item For all 1-cells $f : x \onecell y$ and $\ff : \dmor{\xx}{\yy}{f}$ a displayed invertible 2-cell
	$\thetatheta_1(f) : \dtwo{\thetatheta_0(\xx) \cdot \FF_1(\ff)}{\GG_1(\ff) \cdot \thetatheta_0(\yy)}{\theta_1(f)}$.
\end{itemize}
Here $\theta_1$ denotes the family of invertible 2-cells corresponding to the naturality squares of $\theta$.
Again several coherencies must be satisfied and the precise formulation can be found in the formalization.
The type of displayed pseudotransformatons from $\FF$ to $\GG$ over $\theta$ is denoted by $\disppstrans{\FF}{\GG}{\theta}$.
\end{defi}

\begin{defi}
Suppose that we have displayed bicategories $\D_1$ and  $\D_2$ over $\B_1$ and $\B_2$, displayed pseudofunctors $\FF$ and $\GG$ from $\D_1$ to $\D_2$ over $F$ and $G$ respectively, and displayed pseudotransformations $\thetatheta$ and $\thetatheta'$ from $\FF$ to $\GG$ over $\theta$ and $\theta'$ respectively.
In addition, let $m$ be an invertible modification from $\theta$ to $\theta'$.
Then a \fat{displayed invertible modification} $\mm$ from $\thetatheta$ to $\thetatheta'$ over $m$ consists of a displayed invertible 2-cell
$\mm_2(\xx) :  \dtwo{\thetatheta(\xx)}{\thetatheta'(\xx)}{m(x)}$ for each $x : \B_1$ and $\xx : \D_1(x)$,
In addition, a coherency must be satisfied, which can be found in the formalization.
The type of displayed invertible modifications from $\thetatheta$ to $\thetatheta'$ over $m$ is denoted by $\dispmodif{\thetatheta}{\thetatheta'}{m}$.
\end{defi}

Each of these gadgets has a total version.

\begin{prob}
\label{prob:total}
We have
\begin{enumerate}
	\item Given a displayed pseudofunctor $\FF : \disppsfun{\D_1}{\D_2}{F}$, to construct a pseudofunctor $\total{\FF} : \pseudo(\total{\D_1}, \total{\D_2})$;
	\item Given a displayed pseudotransformation $\thetatheta : \disppstrans{\FF}{\GG}{\theta}$, to construct a pseudotransformation $\total{\thetatheta} : \pstrans{\total{\FF}}{\total{\GG}}$;
	\item Given a displayed invertible modification $\mm : \xymatrix@C=1em{\thetatheta \ar@3[r]^-{m} & \thetatheta',}$ to construct an invertible modificaton $\total{\mm} : \modif{\total{\thetatheta}}{\total{\thetatheta'}}$.
\end{enumerate}
\end{prob}

\begin{construction}{prob:total}
\label{constr:biadj}
By pairing.
\end{construction}

Before we can define displayed biadjunctions, we need several operations on the displayed gadgets we introduced.

\begin{exa}
We have the following
\begin{itemize}
	\item We have $\id(\D) : \disppsfun{\D}{\D}{\id(\B)}$
	where $\id(\B)$ is the identity pseudofunctor;
	\item Given $\FF : \disppsfun{\D_1}{\D_2}{F}$ and $\GG : \disppsfun{\D_2}{\D_3}{G}$,
	we have $\FF \cdot \GG : \disppsfun{\D_1}{\D_3}{F \cdot G}$
	where $F \cdot G$ is the composition of pseudofunctors;
	\item Given $\FF : \disppsfun{\D_1}{\D_2}{F}$,
	we have $\id_1(\FF) : \disppstrans{\FF}{\FF}{\id_1(F)}$
	where $\id_1(F)$ is the identity pseudotransformation on $F$;
	\item Given $\thetatheta : \disppstrans{\FF}{\GG}{\theta}$ and $\thetatheta' : \disppstrans{\GG}{\HH}{\theta'}$,
	we have $\thetatheta \vcomp \thetatheta' : \disppstrans{\FF}{\HH}{\theta \vcomp \theta'}$
	where $\theta \vcomp \theta'$ is the composition of pseudotrasformations;
	\item Given $\FF : \disppsfun{\D_1}{\D_2}{F}$, $\GG : \disppsfun{\D_2}{\D_3}{G}$, $\HH : \disppsfun{\D_2}{\D_3}{H}$, and $\thetatheta : \disppstrans{\GG}{\HH}{\theta}$,
	we have
	\[
	\FF \whiskerl \thetatheta : \disppsfun{\FF \cdot \GG}{\FF \cdot \HH}{F \whiskerl \theta};
	\]
	\item Given $\FF : \disppsfun{\D_1}{\D_2}{F}$, $\GG : \disppsfun{\D_1}{\D_2}{G}$, $\HH : \disppsfun{\D_2}{\D_3}{H}$, and $\thetatheta : \disppstrans{\FF}{\GG}{\theta}$,
	we have
	\[
	\thetatheta \whiskerr \HH : \disppsfun{\FF \cdot \HH}{\GG \cdot \HH}{\theta \whiskerr H};
	\]
	\item Given $\FF : \disppsfun{\D_1}{\D_2}{F}$,
	we have
	\[
	\lunitor{\FF} : \disppstrans{\id \cdot \FF}{\FF}{\lunitor{F}}, \quad
	\runitor{\FF} : \disppstrans{\FF \cdot \id}{\FF}{\runitor{\FF}},
 	\]
	\[
	\linvunitor{\FF} : \disppstrans{\FF}{\id \cdot \FF}{\linvunitor{F}}, \quad
	\rinvunitor{\FF} : \disppstrans{\FF}{\FF \cdot \id}{\rinvunitor{\FF}};
	\]
	\item Given $\FF : \disppsfun{\D_1}{\D_2}{F}$, $\GG : \disppsfun{\D_2}{\D_3}{G}$, and $\HH : \disppsfun{\D_3}{\D_4}{H}$,
	we have
	\[
	\lassoc{\FF}{\GG}{\HH} : \disppstrans{(\FF \cdot \GG) \cdot \HH}{\FF \cdot (\GG \cdot \HH)}{\lassoc{F}{G}{H}},
	\]
	\[
	\rassoc{\FF}{\GG}{\HH} : \disppstrans{\FF \cdot (\GG \cdot \HH)}{(\FF \cdot \GG) \cdot \HH}{\rassoc{F}{G}{H}}.
	\]
\end{itemize}
\end{exa}

\begin{defi}
Suppose we have bicategories $\B_1$ and $\B_2$ and a biadjunction $L \dashv R$ from $\B_1$ to $\B_2$.
We write $\eta$ and $\epsilon$ for the unit and counit of $L \dashv R$ respectively,
and we write $\tau_1$ and $\tau_2$ for the left and right triangle respectively.
Suppose, that we also have displayed bicategories $\D_1$ and $\D_2$ over $\B_1$ and $\B_2$ respectively
and a displayed pseudofunctor $\LL : \disppsfun{\D_1}{\D_2}{L}$.
Then we say $\LL$ is a \fat{displayed left biadjoint pseudofunctor} if we have
\begin{itemize}
	\item A displayed pseudofunctor
	$
	\RR : \disppsfun{\D_2}{\D_1}{R}
	$;
	\item Displayed pseudotransformations
	\[
	\etaeta : \disppstrans{\id}{\LL \cdot \RR}{\eta}, \quad
	\epseps : \disppstrans{\RR \cdot \LL}{\id}{\epsilon};
	\]
	\item Displayed invertible modifications
	\[
	\tautaul : \dispmodif{\rho^{-1} \vcomp \RR \whiskerl \etaeta \vcomp \alpha \vcomp \epseps \whiskerr \RR \vcomp \lambda}{\id_1(\RR),}{\tau_1}
	\]
	\[
	\tautaur : \dispmodif{\lambda^{-1} \vcomp \etaeta \whiskerr \LL \vcomp \alpha^{-1} \vcomp \LL \whiskerl \epseps \vcomp \rho}{\id_1(\LL).}{\tau_2}
	\]
\end{itemize}
\end{defi}

From Construction \ref{constr:biadj}, we get

\begin{prop}
\label{prop:total_biadj}
Given a displayed left biadjoint pseudofunctor $\LL$,
then $\total{\LL}$ is a left biadjoint pseudofunctor.
\end{prop}

\begin{figure}
\[
\xymatrix
{
	\algebra{\sign} \ar@/^/[rr]_-{\top}^-{\algpgrpd} \ar[d] & & \algG{\sign} \ar@/^/[ll]^-{\alggquot} \ar[d] \\
	\pathalg{\sign} \ar@/^/[rr]_-{\top}^-{\pathpgrpd} \ar[d] & & \pathalgG{\sign} \ar@/^/[ll]^-{\pathgquot} \ar[d] \\
	\prealg{\sign} \ar@/^/[rr]_-{\top}^-{\prepgrpd} \ar[d] & & \prealgG{\sign} \ar@/^/[ll]^-{\pregquot} \ar[d]\\
	\onetypes \ar@/^/[rr]_-{\top}^-{\pgrpd} & & \grpd \ar@/^/[ll]^-{\gquot}
}
\]
\caption{The biadjunction}
\label{fig:biadj}
\end{figure}

Now let us use the introduced notions to construct the biadjunction on the level of algebras.
Our approach is summarized in Figure \ref{fig:biadj}.
We start by showing that the groupoid quotient gives rise to a biadjunction.

\begin{prob}
\label{prob:gquit_biadj}
To construct $\gquot \dashv \pgrpd$ where $\gquot : \pseudo(\grpd, \onetypes)$.
\end{prob}

\begin{construction}{prob:gquit_biadj}
\label{constr:gquit_biadj}
We only show how the involved pseudofunctors are defined.
The pseudofunctor $\gquot$ is the groupoid quotient
while $\pgrpd$ sends a 1-type $X$ to the groupoid whose objects are points of $X$ and morphisms from $x$ to $y$ are paths $x = y$.
\end{construction}
Notice that the above biadjunction turns into a biequivalence if we consider the bicategory of univalent groupoids in place of $\grpd$. This implies that the biadjunction of Problem \ref{prob:gquit_biadj} cannot be a biequivalence, since not every groupoid is equivalent to a univalent groupoid.

Next we lift this biadjunction to the level of algebras using the displayed machinery introduced in this section.

\begin{prob}
\label{prob:alg_biadj}
Given a signature $\sign$, to construct a biadjunction $\alggquot \dashv \algpgrpd$ where $\alggquot : \pseudo(\algG{\sign}, \algebra{\sign})$.
\end{prob}

\begin{construction}{prob:alg_biadj}
\label{constr:alg_biadj}
We only give a very brief outline of the construction.

We start by constructing a displayed biadjunction
from $\DFAlg(\semPG{\pointconstr[\sign]})$
to $\DFAlg(\semP{\pointconstr[\sign]})$
over the biadjunction from
Construction \ref{constr:gquit_biadj}.
To do so, we first need to lift the pseudofunctors, and for that, we generalize the approach of Hermida and Jacobs
to the bicategorical setting \cite[Theorem 2.14]{hermida1998structural}.
This requires us to construct two pseudotransformations.
\[
p_1 : \pstrans{\semP{P} \cdot \gquot}{\gquot \cdot \semPG{P}},
\]
\[
p_2 : \pstrans{\semPG{P} \cdot \pgrpd}{\pgrpd \cdot \semP{P}}.
\]
We denote the total biadjunction of the resulting displayed biadjunction by $\pregquot \dashv \prepgrpd$.

Next we lift the biadjunction to the level of path algebras
and for that, we construct a displayed biadjunction between 
$\DCell(\semEG{\pathleft[\sign](i)},\semEG{\pathright[\sign](i)})$
and $\DCell(\semE{\pathleft[\sign](i)},\semE{\pathright[\sign](i)}$
for all $j : \pathlabel$.
Denote the resulting total biadjunction by $\pathgquot \dashv \pathpgrpd$.

To finish the proof, we need to construct one more displayed biadjunction.
For that, we only need to show that if $G : \pathalgG{\sign}$ is an algebra, then $\pathgquot(G)$ also is an algebra,
and if $X : \pathalg{\sign}$ is an algebra, then so is $\pathpgrpd(X)$.
\end{construction}

The next proposition concludes this section.

\begin{prop}
\label{prop:biinitial_in_grpd}
If $G$ is an biinitial object in $\algG{\sign}$,
then $\alggquot(G)$ is a biinitial object in $\algebra{\sign}$.
\end{prop}

\section{HIT Existence}
\label{sec:existence}
From Theorem \ref{thm:initial_alg_sem} we know that initiality implies the induction principle.
Hence, it suffices to construct a biinitial object in the bicategory of algebras in 1-types.
By Proposition \ref{prop:biinitial_in_grpd}, it suffices to construct a biinitial object in $\algG{\sign}$.
To do so, we adapt the semantics by Dybjer and Moeneclaey to our setting \cite{DBLP:journals/entcs/DybjerM18}.

\begin{prob}
\label{prob:initial_grpd_alg}
Given a signature $\sign$, to construct a biinitial object  $\constfont{G}$ in $\algG{\sign}$.
\end{prob}

\begin{construction}{prob:initial_grpd_alg}
\label{constr:initial_grpd_alg}
We only discuss how the carrier $G$ of $\constfont{G}$ is defined.
\begin{itemize}
	\item Note  that each polynomial $P$ gives rise to a container $\hat{P}$.
	Note that each container induces a W-type \cite{abbott2003categories},
	and we define the type of objects of $G$ to be the W-type induced by $\hat{\pointconstr}$.
	Denote this type by $\initob$ and let $\constfont{c}^{\initob} : \pointconstr(\initob) \to \initob$ its algebra map.
	\item The morphisms of $G$ are constructed as a set quotient.
	The main idea is to define the type of morphisms and equalities between them so that the groupoid has the desired structure.
	For each part of the structure, we add a constructor, so concretely, we have constructors witnessing the path constructors, identity, composition, and all other laws.
	First, we define an inductive type $\initmorgen{P}{x}{y}$, for each polynomial $P :\poly$ and elements $x, y : P (\initob)$. Its constructors are given in Figure \ref{fig:initmor}.
    When $P$ is $\idP$, we write $\initmor{x}{y}$ instead of $\initmorgen{\idP}{x}{y}$.
	Afterwards, for each $P:\poly$,$x, y : P(\initob)$ and $f, g : \initmorgen{P}{x}{y}$, we define a type $\initeqgen{P}{f}{g}$.
	Again this type is defined inductively, and its constructors can be found in Figure \ref{fig:initmoreq}.
	Note that given $p : \initmorgen{P}{x}{y}$, we can define $\semE{e}(p) : \initeqgen{Q}{\semE{e}(x)}{\semE{e}(y)}$ where $e$ is an endpoint. When $P$ is $\idP$, we write $\initeq{f}{g}$ instead of $\initeqgen{\idP}{f}{g}$.
	
	Note that the input of the quotient is an equivalence relation, which is valued in propositions.
	For this reason, we define $\initeqprop{f}{g}$ to be the propositional truncation of $\initeq{f}{g}$.
	All in all, we define the morphisms from $x$ to $y$ to be the set quotient of $\initmor{x}{y}$ by $\approx^p$.
	\qedhere	
\end{itemize}
\end{construction}

\begin{figure*}[t]
\begin{center}
\begin{bprooftree}
\AxiomC{$P : \poly$}
\AxiomC{$x : P(\initob)$}
\BinaryInfC{$\initmorid{x} : \initmorgen{P}{x}{x}$}
\end{bprooftree}
\begin{bprooftree}
\AxiomC{$P : \poly$}
\AxiomC{$x,y : P(\initob)$}
\AxiomC{$f : \initmorgen{P}{x}{y}$}
\TrinaryInfC{$\initmorinv{f} : \initmorgen{P}{y}{x}$}
\end{bprooftree}
\end{center}

\vspace{5pt}

\begin{center}
\begin{bprooftree}
\AxiomC{$P : \poly$}
\AxiomC{$x,y,z : P(\initob)$}
\AxiomC{$f : \initmorgen{P}{x}{y}$}
\AxiomC{$g : \initmorgen{P}{y}{z}$}
\QuaternaryInfC{$\initmorcomp{f}{g} : \initmorgen{P}{x}{z}$}
\end{bprooftree}
\end{center}

\vspace{5pt}

\begin{center}
\begin{bprooftree}
\AxiomC{$P,Q:\poly$}
\AxiomC{$x,y : P(\initob)$}
\AxiomC{$f : \initmorgen{P}{x}{y}$}
\TrinaryInfC{$\initmorinl{f} : \initmorgen{P+Q}{\inl(x)}{\inl(y)}$}
\end{bprooftree}
\begin{bprooftree}
\AxiomC{$P,Q:\poly$}
\AxiomC{$x,y : Q(\initob)$}
\AxiomC{$f : \initmorgen{Q}{x}{y}$}
\TrinaryInfC{$\initmorinr{f} : \initmorgen{P+Q}{\inr(x)}{\inr(y)}$}
\end{bprooftree}
\end{center}

\vspace{5pt}

\begin{center}
\begin{bprooftree}
\AxiomC{$P,Q : \poly$}
\AxiomC{$x,y : P(\initob)$}
\AxiomC{$w,z : Q(\initob)$}
\AxiomC{$f : \initmorgen{P}{x}{y}$}
\AxiomC{$g : \initmorgen{Q}{w}{z}$}
\QuinaryInfC{$\initmorpair{f}{g} : \initmorgen{P\times Q}{\pair{x}{w}}{\pair{y}{z}}$}
\end{bprooftree}
\end{center}

\vspace{5pt}

\begin{center}
\begin{bprooftree}
\AxiomC{$j : \pathlabel$}
\AxiomC{$x : \patharg_j(\initob)$}
\BinaryInfC{$\initmorpath{j,x} : \initmorgen{\idP}{\semE{\pathleft_j}(x)}{\semE{\pathright_j}(x)}$}
\end{bprooftree}
\begin{bprooftree}
\AxiomC{$x,y : \pointconstr(\initob)$}
\AxiomC{$f : \initmorgen{\pointconstr}{x}{y}$}
\BinaryInfC{$\initmorap{f} : \initmorgen{\idP}{\constfont{c}^{\initob}(x)}{\constfont{c}^{\initob}(y)}$}
\end{bprooftree}
\end{center}
\caption{Rules for the type $\initmorgen{P}{x}{y}$.}
\label{fig:initmor}
\end{figure*}

\begin{prob}
\label{prob:hit_exist}
Each signature has a HIT.
\end{prob}

\begin{construction}{prob:hit_exist}
\label{constr:hit_exist}
By Propositions \ref{thm:initial_alg_sem} and \ref{prop:biinitial_in_grpd}, it suffices to find a biinitial object in $\algG{\sign}$.
The desired object is given in Construction \ref{constr:initial_grpd_alg}.
\end{construction}

\begin{figure*}
\begin{center}
\begin{bprooftree}
\AxiomC{$f : \initmorgen{P}{x}{y}$}
\UnaryInfC{$\initeqgen{P}{f}{f}$}
\end{bprooftree}
\begin{bprooftree}
\AxiomC{$f,g : \initmorgen{P}{x}{y}$}
\AxiomC{$\initeqgen{P}{f}{g}$}
\BinaryInfC{$\initeqgen{P}{g}{f}$}
\end{bprooftree}
\begin{bprooftree}
\AxiomC{$f,g,h : \initmorgen{P}{x}{y}$}
\AxiomC{$\initeqgen{P}{f}{g}$}
\AxiomC{$\initeqgen{P}{g}{h}$}
\TrinaryInfC{$\initeqgen{P}{f}{h}$}
\end{bprooftree}
\end{center}

\begin{center}
\begin{bprooftree}
\AxiomC{$f : \initmorgen{P}{x}{y}$}
\UnaryInfC{$\initeqgen{P}{\initmorcomp{f}{\initmorid{y}}}{f}$}
\end{bprooftree}
\begin{bprooftree}
\AxiomC{$f : \initmorgen{P}{x}{y}$}
\UnaryInfC{$\initeqgen{P}{\initmorcomp{\initmorid{x}}{f}}{f}$}
\end{bprooftree}
\end{center}

\vspace{5pt}

\begin{center}
\begin{bprooftree}
\AxiomC{$f : \initmorgen{P}{w}{x}$}
\AxiomC{$g : \initmorgen{P}{x}{y}$}
\AxiomC{$h : \initmorgen{P}{y}{z}$}
\TrinaryInfC{$\initeqgen{P}{\initmorcomp{f}{\initmorcomp{f}{g}}}{\initmorcomp{\initmorcomp{f}{g}}{h}}$}
\end{bprooftree}
\end{center}

\vspace{5pt}

\begin{center}
\begin{bprooftree}
\AxiomC{$f : \initmorgen{P}{x}{y}$}
\UnaryInfC{$\initeqgen{P}{\initmorcomp{f}{\initmorinv{f}}}{\initmorid{x}}$}
\end{bprooftree}
\begin{bprooftree}
\AxiomC{$f : \initmorgen{P}{x}{y}$}
\UnaryInfC{$\initeqgen{P}{\initmorcomp{\initmorinv{f}}{f}}{\initmorid{y}}$}
\end{bprooftree}
\end{center}

\vspace{5pt}

\begin{center}
\begin{bprooftree}
\AxiomC{$f,g:\initmorgen{P}{x}{y}$}
\AxiomC{$\initeqgen{P}{f}{g}$}
\BinaryInfC{$\initeqgen{P}{\initmorinv{f}}{\initmorinv{g}}$}
\end{bprooftree}
\begin{bprooftree}
\AxiomC{$f_1,f_2 : \initmorgen{P}{x}{y}$}
\AxiomC{$g : \initmorgen{P}{y}{z}$}
\AxiomC{$\initeqgen{P}{f_1}{f_2}$}
\TrinaryInfC{$\initeqgen{P}{\initmorcomp{f_1}{g}}{\initmorcomp{f_2}{g}}$}
\end{bprooftree}
\end{center}

\vspace{5pt}

\begin{center}
\begin{bprooftree}
\AxiomC{$f : \initmorgen{P}{x}{y}$}
\AxiomC{$g_1,g_2 : \initmorgen{P}{y}{z}$}
\AxiomC{$\initeqgen{P}{g_1}{g_2}$}
\TrinaryInfC{$\initeqgen{P}{\initmorcomp{f}{g_1}}{\initmorcomp{f}{g_2}}$}
\end{bprooftree}
\end{center}

\vspace{5pt}

\begin{center}
\begin{bprooftree}
\AxiomC{$\initeqgen{\idP}{\initmorap{\initmorid{x}}}{\initmorid{\constfont{c}^{\initob}(x)}}$}
\end{bprooftree}
\begin{bprooftree}
\AxiomC{$f,g : \initmorgen{\pointconstr}{x}{y}$}
\AxiomC{$\initeqgen{\pointconstr}{f}{g}$}
\BinaryInfC{$\initeqgen{\idP}{\initmorap{f}}{\initmorap{g}}$}
\end{bprooftree}
\end{center}

\vspace{5pt}

\begin{center}
\begin{bprooftree}
\AxiomC{$f : \initmorgen{\pointconstr}{x}{y}$}
\AxiomC{$g : \initmorgen{\pointconstr}{y}{z}$}
\BinaryInfC{$\initeqgen{\idP}{\initmorap{\initmorcomp{f}{g}}}{\initmorcomp{\initmorap{f}}{\initmorap{g}}}$}
\end{bprooftree}
\end{center}

\vspace{5pt}

\begin{center}
\begin{bprooftree}
\AxiomC{$\initeqgen{P+Q}{\initmorinl{\initmorid{x}}}{\initmorid{\inl{x}}}$}
\end{bprooftree}
\begin{bprooftree}
\AxiomC{$\initeqgen{P+Q}{\initmorinr{\initmorid{x}}}{\initmorid{\inr{x}}}$}
\end{bprooftree}
\end{center}

\vspace{5pt}

\begin{center}
\begin{bprooftree}
\AxiomC{$f,g:\initmorgen{P}{x}{y}$}
\AxiomC{$\initeqgen{P}{f}{g}$}
\BinaryInfC{$\initeqgen{P+Q}{\initmorinl{f}}{\initmorinl{g}}$}
\end{bprooftree}
\begin{bprooftree}
\AxiomC{$f : \initmorgen{P}{x}{y}$}
\AxiomC{$g : \initmorgen{P}{y}{z}$}
\BinaryInfC{$\initeqgen{P + Q}{\initmorinl{\initmorcomp{f}{g}}}{\initmorcomp{\initmorinl{f}}{\initmorinl{g}}}$}
\end{bprooftree}
\end{center}

\vspace{5pt}

\begin{center}
\begin{bprooftree}
\AxiomC{$f,g:\initmorgen{Q}{x}{y}$}
\AxiomC{$\initeqgen{Q}{f}{g}$}
\BinaryInfC{$\initeqgen{P+Q}{\initmorinr{f}}{\initmorinr{g}}$}
\end{bprooftree}
\begin{bprooftree}
\AxiomC{$f : \initmorgen{Q}{x}{y}$}
\AxiomC{$g : \initmorgen{Q}{y}{z}$}
\BinaryInfC{$\initeqgen{P + Q}{\initmorinr{\initmorcomp{f}{g}}}{\initmorcomp{\initmorinr{f}}{\initmorinr{g}}}$}
\end{bprooftree}
\end{center}

\vspace{5pt}

\begin{center}
\begin{bprooftree}
\AxiomC{$\initeqgen{P\times Q}{\initmorpair{\initmorid{x}}{\initmorid{y}}}{\initmoridvar{\pair{x}{y}}}$}
\end{bprooftree}
\end{center}

\vspace{5pt}

\begin{center}
\begin{bprooftree}
\AxiomC{$f_1,f_2:\initmorgen{P}{x_1}{x_2}$}
\AxiomC{$g_1,g_2:\initmorgen{Q}{y_1}{y_2}$}
\AxiomC{$\initeqgen{P}{f_1}{f_2}$}
\AxiomC{$\initeqgen{Q}{g_1}{g_2}$}
\QuaternaryInfC{$\initeqgen{P\times Q}{\initmorpair{f_1}{g_1}}{\initmorpair{f_2}{g_2}}$}
\end{bprooftree}
\end{center}

\vspace{5pt}

\begin{center}
\begin{bprooftree}
\AxiomC{$f_1 : \initmorgen{P}{x_1}{y_1}$}
\AxiomC{$g_1 : \initmorgen{P}{y_1}{z_1}$}
\AxiomC{$f_2 : \initmorgen{Q}{x_2}{y_2}$}
\AxiomC{$g_2 : \initmorgen{Q}{y_2}{z_2}$}
\QuaternaryInfC{$\initeqgen{P \times Q}{\initmorpair{\initmorcomp{f_1}{g_1}}{\initmorcomp{f_2}{g_2}}}{\initmorcomp{\initmorpair{f_1}{f_2}}{\initmorpair{g_1}{g_2}}}$}
\end{bprooftree}
\end{center}

\vspace{5pt}

\begin{center}
\begin{bprooftree}
\AxiomC{$j : \homotlabel$}
\AxiomC{$x : \homotpointarg_j(G_0)$}
\AxiomC{$p : \initmorgen{\homotpathtarg}{\semE{\pathargleft}(x)}{\semE{\pathargright}(x)}$}
\TrinaryInfC{$\initeqgen{\idP}{\semH{\homotleft_j}(x, p)}{\semH{\homotright_j}(x, p)}$}
\end{bprooftree}
\end{center}

\vspace{5pt}

\begin{center}
\begin{bprooftree}
\AxiomC{$j : \pathlabel$}
\AxiomC{$x, y : \patharg_j(G_0)$}
\AxiomC{$p : \initmorgen{\patharg_j}{x}{y}$}
\TrinaryInfC{$\initeqgen{\idP}{\initmorcomp{\initmorpath{j, x}}{\semE{\pathright}(p)}}{\initmorcomp{\semE{\pathleft}{(p)}}{\initmorpath{j, y}}}$}
\end{bprooftree}
\end{center}

\caption{Rules for the type $\initeqgen{P}{f}{g}$.}
\label{fig:initmoreq}
\end{figure*}

\section{Additional Examples}
\label{sec:examples}
In this section we present some additional examples, which complement the
ones introduced in Section \ref{sec:signs}. Remember that our higher
inductive types are all 1-truncated, and we always omit the 1-truncation
constructor from their syntax.

\subsection{Coinserter}
\label{sec:coinserter}

The coinserter is a bicategorical generalization of the coequalizer in
a category. In the bicategory of 1-types, coinserters can be constructed as
 homotopy coequalizer \cite[Chapter 6]{hottbook}.

\begin{defi}
Let $\B$ be a bicategory. Let $A$ and $B$ be objects of $\B$ and let
$f,g : A \onecell B$. The \fat{coinserter} of $f$ and $g$ is an
object $Q$ together with a 1-cell $q : B \onecell Q$ and a 2-cell
$\theta : f \cdot q\twocell g \cdot q$.
The triple $(Q,q,\theta)$ must satisfy the following universal
property. Suppose that we have
\begin{itemize}
\item an object $Q'$;
\item a 1-cell $q' : B \onecell Q'$;
\item a 2-cell $\theta' : f \cdot q' \twocell g \cdot q'$.
\end{itemize}
Then there exists a 1-cell $h : Q \onecell Q'$ together with a 2-cell
$\phi : q \cdot h \twocell q'$ with a path $\theta \whiskerr h \vcomp
g \whiskerl \phi = \lassoc{f}{q}{h} \vcomp f \whiskerl \phi \vcomp
\theta'$. The pair $(h,\phi)$ is unique up to unique 2-cell, which means that
given another 1-cell $h' : Q \onecell Q'$ and another 2-cell $\phi' :
q \cdot h' \twocell q'$ with a path $\theta \whiskerr h' \vcomp g
\whiskerl \phi' = \lassoc{f}{q}{h'} \vcomp f \whiskerl \phi' \vcomp
\theta'$, then there exists a unique 2-cell $\tau : h
\twocell h'$ such that $q \whiskerl \tau \vcomp \phi' = \phi$.
\end{defi}

Next we show how to construct coinserters in the bicategory of 1-types.  Given two 1-types $A$ and $B$ and two
functions $f , g: A \to B$, the coinserter of $f$ and $g$ is
given by the following HIT:
\begin{lstlisting}[mathescape=true]
Inductive $\coequalizer f g$ :=
| $\coequalizerbase$ : $B \rightarrow \coequalizer f g$
| $\coequalizerglue$ : $\depprod{(x : A)}{\coequalizerbase(f(x)) = \coequalizerbase(g(x))}$
\end{lstlisting}

Here are all the ingredients needed to specify the signature $\coequalizer f g$ for the coinserter:
\begin{itemize}
\item $\pointconstr[\coequalizer f g] \eqdef \constantP{B}$;
\item $\pathlabel[\coequalizer f g] \eqdef \unit$, and for its unique inhabitant take $\patharg[\coequalizer f g] \eqdef \constantP{A}$ and endpoints
  \[
  \pathleft[\coequalizer f g] \eqdef \comp{\fmap(f)}{\constr}, \quad
  \pathright[\coequalizer f g] \eqdef \comp{\fmap(g)}{\constr}; 
  \]
\item $\homotlabel[\coequalizer f g]$ is the empty type.
\end{itemize}


\subsection{Coequifier}
\label{sec:coequifier}

The coequifier is a finite colimit in a bicategory, corresponding to a
higher version of the coequalizer.
While the coinserter makes a diagram of 1-cell commute up to a 2-cell,
the coequifier makes a diagram of 2-cells commute strictly.

\begin{defi}
Let $\B$ be a bicategory. Let $A$ and $B$ be objects of $\B$, let $f,
g : A \onecell B$ be 1-cells and $\beta,\gamma : f \twocell
g$ be 2-cells. The \fat{coequifier} of $\beta$ and $\gamma$ is an object $Q$
together with a 1-cell $q : B \onecell Q$ and a path
$\beta \whiskerr q = \gamma \whiskerr q$.
The pair $(Q,q)$ must satisfy the following universal
property. Suppose that we have
\begin{itemize}
\item an object $Q'$,
\item a 1-cell $q' : B \onecell Q'$, and
\item a path $\beta \whiskerr q' = \gamma \whiskerr q'$.
\end{itemize}
Then there exists a 1-cell $h : Q \onecell Q'$ together with a 2-cell
$\phi : q \cdot h \twocell q'$. The pair $(h,\phi)$ is unique up to a unique 2-cell, which means that
given another 1-cell $h' : Q \onecell Q'$ and another 2-cell $\phi' :
q \cdot h' \twocell q'$, there exists a unique 2-cell $\tau : h
\twocell h'$ such that $q \whiskerl \tau \vcomp \phi' = \phi$.
\end{defi}

In the bicategory of 1-types, the coequifier can also be constructed as a HIT. In
the definition below, $A$ and $B$ are 1-types, $f , g: A \to B$ are
functions and $\beta$ and $\gamma$ are homotopies between $f$ and $g$.
\begin{lstlisting}[mathescape=true]
Inductive $\coequifier \beta \gamma$ :=
| $\coequifierbase$ : $B \rightarrow \coequifier \beta \gamma$
| $\coequifierglue$ : $\depprod{(x : A)}{\ap {\coequifierbase}{(\beta(x))} = \ap {\coequifierbase}{(\gamma(x))}}$
\end{lstlisting}

The signature $\coequifier \beta \gamma$ for the coequifier of $\beta$
and $\gamma$ is given as follows:
\begin{itemize}
\item $\pointconstr[\coequifier f g] \eqdef \constantP{B}$;
\item $\pathlabel[\coequifier f g]$ is the empty type; 
\item $\homotlabel[\coequifier f g] \eqdef \unit$, and take 
  $\homotpointarg[\coequifier f g] \eqdef \constantP{A}$.
  The homotopy constructor $\coequifierglue$ does not have path arguments, so we take  
  $\homotpathtarg[\coequifier f g] \eqdef \constantP{\unit}$ and
  $\pathargleft[\coequifier f g] \eqdef \pathargright[\coequifier f g] \eqdef \Ce(\unitt)$.
  The endpoints $\homotpathleft[\coequifier f g]$ and $\homotpathright[\coequifier f g]$ are:
  \[
  \homotpathleft[\coequifier f g] \eqdef \comp{\fmap(f)}{\constr}, \quad
  \homotpathright[\coequifier f g] \eqdef \comp{\fmap(g)}{\constr},
  \]
  while the left and right homotopy endpoints are:
  \[
  \hap{\constr}{(\idtoH(\ap{\fmap}{(\funextsec(\beta))}))}, \quad
  \hap{\constr}{(\idtoH(\ap{\fmap}{(\funextsec(\gamma))}))}.
  \]
\end{itemize}
In the construction of the homotopy endpoints we used the function
$\idtoH$ introduced in Section~\ref{sec:signatures}, which embeds
paths between endpoints into homotopy endpoints. Notice also the
difference between $\constructor{ap}$ and $\constfont{ap}$: the first
is an homotopy endpoint constructor, the second indicates the
application of a function to a path.

\subsection{Group Quotient}
\label{sec:group_quotient}

Now we introduce a particular instance of the groupoid quotient, that
we call the \emph{group quotient}. We start with a group $G$ and we write $\id$ 
for its unit and $\cdot$ for multiplication in $G$. Define a groupoid $\widehat{G}$
with only one object and with $G$ as the only homset. The group
quotient of $G$ is the groupoid quotient of $\widehat{G}$, and it
corresponds to the following HIT:
\begin{lstlisting}[mathescape=true]
Inductive $\groupquot G$ :=
| $\groupquotbase$ : $\groupquot G$
| $\groupquotloop$ : $G \rightarrow \groupquotbase = \groupquotbase$
| $\groupquotloope$ : $\groupquotloop(\id)= \refl \groupquotbase$
| $\groupquotloopm$ : $\depprod{(x,y : G)}{\groupquotloop(x \cdot y)= \concat{\groupquotloop(x)}{\groupquotloop(y)}}$
\end{lstlisting}

The signature $\groupquot G$ for the group quotient is defined as follows:
\begin{itemize}
\item $\pointconstr \eqdef \constantP{\unit}$;
\item $\pathlabel \eqdef \unit$, and for its unique inhabitant take $\patharg \eqdef \constantP{G}$ and both endpoints $\pathleft$ and $\pathright$ equal to $\comp{\Ce(\unitt)}{\constr}$.
\item $\homotlabel \eqdef \bool$, where $\bool$ is the type of
booleans with inhabitants $\booltrue$ and $\boolfalse$. This means
that there are two homotopy constructors: $\groupquotloope$, with
associated label $\booltrue$, and $\groupquotloopm$, with associated
label $\boolfalse$.
\item The constructor $\groupquotloope$ does not have
point arguments, so we take $\homotpointarg_{\booltrue}
\eqdef \constantP{\unit}$. It also does not have path arguments, therefore
$\homotpathtarg_{\booltrue} \eqdef \constantP{\unit}$ and
$\pathargleft_{\booltrue} \eqdef \pathargright_{\booltrue} \eqdef \Ce(\unitt)$.
The endpoints $\homotpathleft_{\booltrue}$ and
$\homotpathright_{\booltrue}$ are both equal to
$\comp{\Ce(\unitt)}{\constr}$. The left homotopy endpoint is
\[
\hconcat{\hap{\constr}{\hinv{(\hcompconst{\Ce(\id)})}}}{
\hconcat{\hinv{\hassocN}}{
\hconcat{\hconstr{}{\Ce(\id)}}{
\hconcat{\hassocN}{\hap{\constr}{(\hcompconst{\Ce(\id)})}}
}}},
\]
while the right homotopy endpoint is $\hrefl{\constr}$.
\item The constructor $\groupquotloopm$ has two point
arguments of type $G$, so we take $\homotpointarg_{\boolfalse}
\eqdef \constantP{G \times G}$. It does not have path arguments, therefore
$\homotpathtarg_{\boolfalse} \eqdef \constantP{\unit}$ and
$\pathargleft_{\boolfalse} \eqdef \pathargright_{\boolfalse}
\eqdef \Ce(\unitt)$.  The endpoints $\homotpathleft_{\boolfalse}$ and
$\homotpathright_{\boolfalse}$ are both equal to
$\comp{\Ce(\unitt)}{\constr}$. The left homotopy endpoint is
\begin{align*}
& \hconcat{\hap{\constr}{\hinv{(\hcompconst{\fmap(\Lam {x,y}. x \cdot y)})}}}{\hinv{\hassocN}} \\
& \quad \hconcat{}{\hconstr{}{\fmap(\Lam {x,y}. x \cdot y)}} \\
& \quad \hconcat{}{\hconcat{\hassocN}{\hap{\constr}{(\hcompconst{\fmap(\Lam {x,y}. x \cdot y)})}}},
\end{align*}
while the right endpoint is
\begin{align*}
& \hconcat{\hap{\constr}{\hinv{(\hcompconst{\fmap(\projl)})}}}{
\hinv{\hassocN}} \\
& \quad \hconcat{}{\hconstr{}{\fmap(\projl)}} \\
& \quad \hconcat{}{\hconcat{\hassocN}{\hap{\constr}{(\hcompconst{\fmap(\projl)})}}} \\
& \quad \hconcat{}{\hconcat{\hap{\constr}{\hinv{(\hcompconst{\fmap(\projr)})}}}{\hinv{\hassocN}}} \\
& \quad \hconcat{}{\hconstr{}{\fmap(\projr)}} \\
& \quad \hconcat{}{\hconcat{\hassocN}}{\hap{\constr}{(\hcompconst{\fmap(\projr)})}}
\end{align*}
\end{itemize}

The signature for the groupoid quotient is obtainable as a slight
generalization of the signature for the group quotient. We do not show
the more general construction here, since this is not conceptually
more enlightening than the (already quite complicated) signature for
the group quotient.

\subsection{Monoidal Object}
\label{sec:monoidal_object}
Next we look at two other examples of signatures.
Here we are not interested in the HIT described by the signature, but instead, we are interested in the algebras.
We first discuss the signature whose algebras are \emph{monoidal objects}.
\begin{lstlisting}[mathescape=true]
Inductive $\monobj$ :=
| $\monobjunit$ : $\monobj$
| $\monobjtensor$ : $\monobj \rightarrow \monobj \rightarrow \monobj$
| $\monobjlambda$ : $\depprod{(x : \monobj)}{\monobjtensor(\monobjunit,x) = x}$
| $\monobjrho$ : $\depprod{(x : \monobj)}{\monobjtensor(x,\monobjunit) = x}$
| $\monobjalpha$ : $\depprod{(x ,y,z: \monobj)}{\monobjtensor(x,(\monobjtensor(y,z)))=\monobjtensor(\monobjtensor(x,y),z)}$
| $\monobjtr$ : $\depprod{(x,y : \monobj)}{\ap{(\Lam z. \monobjtensor(x,z))}{(\monobjlambda(y))}} = \concat{\monobjalpha(x,\monobjunit,y)}{\ap{(\Lam z. \monobjtensor(z,y))}{(\monobjrho(x))}}$
| $\monobjpent$ : $\depprod{(w,x,y,z : \monobj)}{}$
    $\concat{\monobjalpha(w,x,\monobjtensor(y,z))}{\monobjalpha(\monobjtensor(w,x),y,z)}$
    =
    $\concat{\ap{(\Lam v. \monobjtensor(w,v))}{(\monobjalpha(x,y,z))}}{\concat{\monobjalpha(w,\monobjtensor(x,y),z)}{\ap{(\Lam v. \monobjtensor(v,z))}{(\monobjalpha(w,x,y))}}}$
\end{lstlisting}
We do not show the signature associated to this HIT here. We redirect
the interested reader to our formalization for the complete
definition.

In the constructors of $\monobj$, one can recognize the data of a
monoidal category. The point constructors $\monobjunit$ and
$\monobjtensor$ correspond to unit object and tensor. The path
constructors $\monobjlambda$, $\monobjrho$ and $\monobjalpha$ are left
unitor, right unitor and associator respectively, while the homotopy constructors
$\monobjtr$ and $\monobjpent$ are the two coherence laws of
monoidal categories. And in fact, algebras in groupoids of the
monoidal object signature are precisely
\emph{monoidal groupoids}, the groupoid variant of monoidal
categories, and $\monobj$ is a presentation of the initial
monoidal groupoid.
Note that Piceghello proved coherence for monoidal groupoids \cite{Piceghello19}, and algebras for the signature $\monobj$ correspond to those monoidal groupoids.

\begin{exa}
Let $A$ be a 1-type.
We can construct an algebra of $\monobj$ whose carrier is the type $\List(A)$ of lists of $A$.
The unit is the empty list, the tensor is concatenation, and the laws and coherencies are proven by induction.
\end{exa}


We also define a signature $\cohgroup$ whose algebras are \emph{coherent 2-groups} \cite{baez2004groups}.
Its definition includes all the constructors of the monoidal object, plus a new point constructor
$\cohgroupinv$ : $\cohgroup \rightarrow \cohgroup$, two new path
constructors
\[
\cohgrouplinv : \depprod{(x : \cohgroup)}{\cohgrouptensor(\cohgroupinv(x),x) = \cohgroupunit},
\quad\quad
\cohgrouprinv : \depprod{(x : \cohgroup)}{\cohgroupunit = \cohgrouptensor(x,\cohgroupinv(x))},
\]
and two new homotopy constructors similar to the coherencies given by Baez and Lauda \cite{baez2004groups}.

We look at two examples of coherent 2-groups.
The first is based on the work of Buchholtz \etal \ \cite{DBLP:conf/lics/BuchholtzDR18} and Kraus and Altenkirch~\cite{KA18}.
They define higher groups as loop spaces and for 1-truncated types, the loop space is a coherent 2-group in our sense.

\begin{exa}
Suppose, $A$ is a 2-type and $a$ is a point of $A$.
Then we can construct an algebra of $\cohgroup$ whose carrier is given by $a = a$.
\end{exa}

The second example is the automorphism group on a 1-type whose elements are equivalences on a given 1-type.
Note that this group is the loop space of the 2-type of 1-types due to univalence.

\begin{exa}
Let $A$ be a 1-type.
We have an algebra of $\cohgroup$ whose carrier is given by equivalences $f : A \to A$.
The unit element is the identity function, the tensor is the concatenation, and the inverse is just the inverse of an equivalence.
\end{exa}

\section{PIE Limits of Algebras}
\label{sec:finite_limits}
This section is dedicated to the construction of PIE limits in the
bicategory $\alg(\sign)$ of algebras in 1-types for the signature
$\sign$ \cite{power1991characterization}.
Note that $\alg(\sign)$ also has a terminal object whose carrier is the unit type.

\subsection{Products}

Binary products in a bicategory generalize the notion of binary
product in a category.

\begin{defi}\label{def:product}
Let $\B$ be a bicategory and let $A$ and $B$ be two objects of
$\B$. The \fat{product} of $A$ and $B$ is given by an object $\prodB
A B$ together with 1-cells $\projlB : \prodB A B \onecell A$ and
$\projrB : \prodB A B \onecell B$.

The triple $(\prodB A B,\projlB, \projrB)$ must satisfy the following
universal property. Given an object $X$ and 1-cells $f : X
\onecell A$ and $g : X \onecell B$, there exist a 1-cell
$\mappair f g : X \onecell \prodB A B$ and two 2-cells $\theta :
\mappair f g \cdot \projlB \twocell f$ and $\theta' : \mappair f g
\cdot \projrB \twocell g$.

Moreover, given two 1-cells $h_1,h_2 : X \onecell \prodB A B$ with
2-cells $\theta_i : h_i \cdot \projlB \twocell f$ and $\theta'_i : h_i
\cdot \projrB \twocell g$, for $i= 1,2$, there exists a unique
2-cell $\tau : h_1 \twocell h_2$ such that $\tau \whiskerr
\projlB \vcomp \theta_2 = \theta_1$ and $\tau \whiskerr \projrB \vcomp
\theta'_2 = \theta'_1$.
\end{defi} 

\begin{prob}
\label{prob:product}
Given algebras $A$ and $B$ for $\sign$, to construct their product in $\alg(\sign)$.
\end{prob}

\begin{construction}{prob:product}\label{cons:product}
Let $A$ and $B$ be two algebras for $\sign$. The product of $A$ and
$B$ consists of the following data:
\begin{itemize}
\item The carrier is $A \times B$, the product of the carriers of $A$ and $B$.
\item The function $\AlgPoint{A \times B} : \polyAct{\pointconstr}{A \times B} \rightarrow A \times B$ is
\[
\AlgPoint{A \times B}(x) \eqdef (\AlgPoint{A}(\polyAct{\pointconstr}{\projl}(x)),\AlgPoint{B}(\polyAct{\pointconstr}{\projr}(x))).
\]
\item For all labels $j : \pathlabel$ we are given pseudonatural
transformations $\semE{\pathleft_j}$ and $\semE{\pathright_j}$. We
write $\semE{\pathleft_j}_X$ for the component of $\semE{\pathleft_j}$
at object $X$. Given a
1-cell $f : X \onecell Y$ in $\alg(\sign)$, we write
$\semE{\pathleft_j}(f)$ for the 2-cell of type
$\semE{\pathleft_j}_X \cdot
f \twocell \polyAct{\patharg_j}{f} \cdot \semE{\pathleft_j}_Y$ (this is the same notation used in Example~\ref{ex:DCell}).
We write similarly for $\semE{\pathright_j}$.

For each point $x : \polyAct{\patharg_j}{A \times B}$, we are required
to construct a path $\AlgPath{A \times B}{j}(x)
: \semE{\pathleft_j}_{A \times B}(x) = \semE{\pathright_j}_{A \times
B}(x)$. This is a path in $A \times B$, so it is enough to construct
two paths $\projl(\semE{\pathleft_j}_{A \times B}(x))
= \projl(\semE{\pathright_j}_{A \times B}(x))$ and
$\projr(\semE{\pathleft_j}_{A \times B}(x))
= \projr(\semE{\pathright_j}_{A \times B}(x))$. The first of these is
defined as the following concatenation of paths:
\begin{align*}
\projl(\semE{\pathleft_j}_{A \times B}(x))
&\stackrel{\semE{\pathleft_j}(\projl)(x)}{=}
\semE{\pathleft_j}_{A}(\polyAct{\patharg_j}{\projl}(x)) \\
&\stackrel{\AlgPath{A}{j}(\polyAct{\patharg_j}{\projl}(x))}{=}
\semE{\pathright_j}_{A}(\polyAct{\patharg_j}{\projl}(x)) \\
&\stackrel{\inverse{(\semE{\pathright_j}(\projl)(x))}}{=}
\projl(\semE{\pathright_j}_{A \times B}(x))
\end{align*}
The second path is defined analogously.
\item The construction of the required homotopies is more involved and we refer the reader to the formalization for all the details.
\end{itemize}

It is not difficult to show that the projections $\projlB : \prodB A
B \onecell A$ and $\projrB : \prodB A B \onecell B$ are morphisms of
algebras.  Moreover, the product of algebras satisfies the required
universal property of Definition~\ref{def:product}. Given
two algebra morphisms $f : X \onecell A$ and $g : X \onecell B$, we
have a function $\mappair f g : X \to A \times B$ by the universal
property of the product (of types). We refer the reader to the formalization for the proof that
$\mappair f g$ is  a morphism of algebras. 
\end{construction}

\subsection{Inserter}
The inserter in a bicategory is a generalization of the equalizer in
a category.

\begin{defi}\label{def:inserter}
Let $\B$ be a bicategory. Let $A$ and $B$ be objects of $\B$ and let
$f,g : A \onecell B$. The \fat{inserter} of $f$ and $g$ is an
object $E$ together with a 1-cell $e : E \onecell A$ and a 2-cell
$\epsilon : e \cdot f\twocell e \cdot g$.
The triple $(E,e,\epsilon)$ must satisfy the following universal
property. Suppose we have
\begin{itemize}
\item an object $E'$;
\item a 1-cell $e' : E' \onecell A$;
\item a 2-cell $\epsilon' : e' \cdot f \twocell e' \cdot g$.
\end{itemize}
Then there exists a 1-cell $h : E' \onecell E$ together with a 2-cell
$\phi : h \cdot e \twocell e'$ and a path $\rassoc{h}{e}{f} \vcomp
h \whiskerl \epsilon \vcomp \lassoc{h}{e}{g} \vcomp \phi \whiskerr g
= \phi \whiskerr f \vcomp \epsilon'$.
The pair $(h,\phi)$ is unique up to unique 2-cell, which means that
given another 1-cell $h' : E' \onecell E$, another 2-cell $\phi' :
h' \cdot e \twocell e'$, and a path $\rassoc{h'}{e}{f} \vcomp
h' \whiskerl \epsilon \vcomp \lassoc{h'}{e}{g} \vcomp \phi' \whiskerr g
= \phi' \whiskerr f \vcomp \epsilon'$, there exists a unique 2-cell $\tau : h
\twocell h'$ such that $\tau \whiskerr e \vcomp \phi' = \phi$.
\end{defi}

Note that there is a choice for the direction of the 2-cell $\epsilon : e \cdot f\twocell e \cdot g$ since we could also have chosen $\epsilon : e \cdot g\twocell e \cdot f$ or we could have required this 2-cell to be inevrtible.
However, since we only consider bicategories in which all 2-cells are invertible, this choice does not matter.

\begin{prob}
\label{prob:inserter}
Given algebras $A$ and $B$ for $\sign$ and
algebra morphisms $f,g : A \onecell B$, to construct the inserter of $f$ and $g$ in $\alg(\sign)$.
\end{prob}

\begin{construction}{prob:inserter}\label{cons:inserter}
The inserter of $f$ and $g$ is defined as the total algebra $\TotalAlg Y$ of a
displayed algebra $Y$ over $A$. Displayed algebras were
introduced in Definition~\ref{def:disp_alg}. $Y$ is defined as follows:
\begin{itemize}
\item The underlying family $Y$ of 1-types over $A$ is $Y(x) \eqdef f(x) = g(x)$.
\item
For each $x : \polyAct{\pointconstr}{A}$, we are required to construct
a map $\DispAlgPoint{Y} : \polyDact{\pointconstr}{Y}(x) \rightarrow
f(\AlgPoint{A}(x)) = g (\AlgPoint{A}(x))$. Suppose we have $\pover{x}
: \polyDact{\pointconstr}{Y}(x)$. By induction on the polynomial
$\pointconstr$, it is possible to derive from $\pover{x}$ a path $p
: \polyAct{\pointconstr}{f}(x) = \polyAct{\pointconstr}{g}(x)$.  We
define $\DispAlgPoint{Y}(\pover{x})$ as the following concatenation of paths:
\[
f(\AlgPoint{A}(x))
\stackrel{\AlgPoint{f}(x)}{=} \AlgPoint{B}(\polyAct{\pointconstr}{f}(x))
\stackrel{\ap{\AlgPoint{B}}{p}}{=} \AlgPoint{B}(\polyAct{\pointconstr}{g}(x))
\stackrel{\inverse{(\AlgPoint{g}(x))}}{=} g (\AlgPoint{A}(x))
\]
\item The construction of paths $\DispAlgPath{Y}{j}$ is relatively involved, and we refer to the formalization for the details.
\item
The construction of globes $\DispAlgHomot{Y}{j}$ is
straightforward. These are paths between paths in $Y(x)$, for some
point $x$. Since $Y(x)$ is a set, the required paths exist.
\end{itemize}

The 1-cell $e$ is the first projection out of the total algebra
$\TotalAlg Y$, which is an algebra morphism by construction. The
2-cell $\epsilon$ is the second projection out of $\TotalAlg Y$. The
algebra $\TotalAlg Y$ satisfies the required universal property of the
inserter spelled out in Definition~\ref{def:inserter}.
\end{construction}

\subsection{Equifier}

The equifier is finite limit in a bicategory, corresponding to a
higher version of the equalizer.

\begin{defi}\label{def:equifier}
Let $\B$ be a bicategory. Let $A$ and $B$ be objects of $\B$, let $f,
g : A \onecell B$ and $\beta,\gamma : f \twocell
g$. The \fat{equifier} of $\beta$ and $\gamma$ is an object $E$
together with a 1-cell $e : E \onecell A$ and a path
$e \whiskerl \beta = e \whiskerl \gamma$.
The pair $(E,e)$ must satisfy the following universal
property. Suppose we have
\begin{itemize}
\item an object $E'$;
\item a 1-cell $e' : E' \onecell A$;
\item a path $e' \whiskerl \beta = e' \whiskerl \gamma$.
\end{itemize}
Then there exists a 1-cell $h : E' \onecell E$ together with a 2-cell
$\phi : h \cdot e \twocell e'$. The pair $(h,\phi)$ is unique up to unique 2-cell, which means that
given another 1-cell $h' : E' \onecell E$ and another 2-cell $\phi' :
h' \cdot e \twocell e'$, there exists a unique 2-cell $\tau : h
\twocell h'$ such that $\tau \whiskerr e \vcomp \phi' = \phi$.
\end{defi}

\begin{prob}
\label{prob:equifier}
Given algebras $A$ and $B$ for $\sign$, given 1-cells $f,g :
A \onecell B$ and 2-cells $\beta,\gamma : f \twocell g$, to construct the
equifier of $\beta$ and $\gamma$ in $\alg(\sign)$.
\end{prob}

\begin{construction}{prob:equifier}\label{cons:equifier}
Similar to the construction of the inserter, we define the equifier
of $\beta$ and $\gamma$ as the total algebra $\TotalAlg Y$ of the following
displayed algebra $Y$ over $A$.
\begin{itemize}
\item The underlying family $Y$ of 1-types over $A$ is $Y(x) \eqdef \beta(x) = \gamma(x)$.
\item
For each $x : \polyAct{\pointconstr}{A}$, we are required to construct
a map $\DispAlgPoint{Y} : \polyDact{\pointconstr}{Y}(x) \rightarrow
\beta(\AlgPoint{A}(x)) = \gamma (\AlgPoint{A}(x))$. Assume given $\pover{x}
: \polyDact{\pointconstr}{Y}(x)$. By induction on the polynomial
$\pointconstr$, it is possible to derive from $\pover{x}$ a path $p
: \polyAct{\pointconstr}{\beta}(x) = \polyAct{\pointconstr}{\gamma}(x)$.  We
define $\DispAlgPoint{Y}(\pover{x})$ as the following concatenation of
paths:
\begin{align*}
\beta(\AlgPoint{A}(x))
&= \hconcat{\beta(\AlgPoint{A}(x))}{\refl{g(\AlgPoint{A}(x))}} \\
&= \hconcat{\beta(\AlgPoint{A}(x))}{(\hconcat{\AlgMapPoint{g}(x)}{\inverse{(\AlgMapPoint{g}(x))}})} \\
&= \hconcat{(\hconcat{\beta(\AlgPoint{A}(x))}{\AlgMapPoint{g}(x)})}{\inverse{(\AlgMapPoint{g}(x))}} \\
&= \hconcat{(\hconcat{\AlgMapPoint{f}(x)}{\ap{\AlgPoint{B}}{(\polyAct{\pointconstr}{\beta}(x))}})}{\inverse{(\AlgMapPoint{g}(x))}}
& \text{(by $\AlgCellPoint{\beta}(x)$)}\\
&= \hconcat{(\hconcat{\AlgMapPoint{f}(x)}{\ap{\AlgPoint{B}}{(\polyAct{\pointconstr}{\gamma}(x))}})}{\inverse{(\AlgMapPoint{g}(x))}}
& \text{(by $p$)}\\
&= \hconcat{(\hconcat{\gamma(\AlgPoint{A}(x))}{\AlgMapPoint{g}(x)})}{\inverse{(\AlgMapPoint{g}(x))}}
& \text{(by $\AlgCellPoint{\gamma}(x)$)}\\
&= \hconcat{\gamma(\AlgPoint{A}(x))}{(\hconcat{\AlgMapPoint{g}(x)}{\inverse{(\AlgMapPoint{g}(x))}})} \\
&= \hconcat{\gamma(\AlgPoint{A}(x))}{\refl{g(\AlgPoint{A}(x))}} \\
&= \gamma(\AlgPoint{A}(x))
\end{align*}
\item The construction of paths $\DispAlgPath{Y}{j}$ and globes $\DispAlgHomot{Y}{j}$ is
straightforward. These are respectively paths and paths between paths
in $Y(x)$ for some point $x$. Since $Y(x)$ is a proposition,
these constructions are all trivial.
\end{itemize}

The 1-cell $e$ is the first projection out of the total algebra $\TotalAlg Y$, which
is an algebra morphism by construction. The 2-cell $\epsilon$ is the
second projection out of $\TotalAlg Y$. The algebra $\TotalAlg Y$ satisfies the required
universal property of the equifier spelled out in Definition~\ref{def:equifier}.
\end{construction}

\section{The Free Algebra}
\label{sec:free_algebra}
In this section, we discuss the free algebra for a signature $\sign$ and show that it gives rise to a left biadjoint pseudofunctor from 1-types to algebras in 1-types for $\sign$.
From this, we conclude that each signature generates a pseudomonad on the bicategory of 1-types \cite{LACK2000179}.

We construct the free algebra for $\sign$ as a biinitial algebra for a modified version of $\sign$.
More specifically, suppose that we have a signature $\sign$ and a 1-type $A$.
We first construct another signature, called the \emph{free signature}, which has all the constructors of $\sign$ and an additional point constructor with arguments from $A$.
Then we define the free $\sign$-algebra on $A$ to be the biinitial algebra for the free signature.

\begin{defi}\label{def:free_alg}
Let $\sign$ be a signature and let $A$ be a 1-type.
Define the \fat{free signature} $\freesign{\sign}{A}$ as follows
\begin{itemize}
	\item $\pointconstr[\freesign{\sign}{A}] \eqdef \sumP{\pointconstr[\sign]}{\constantP{A}}$
	\item Note that each path endpoint $e : \pathendpoint{\pointconstr[\sign]}{P}{Q}$ gives rise to another path endpoint $\freepath{e} : \pathendpoint{\pointconstr[\freesign{\sign}{A}]}{P}{Q}$.
	We choose $\pathlabel[\freesign{\sign}{A}] \eqdef\pathlabel[\sign]$ and $\patharg[\freesign{\sign}{A}] \eqdef \patharg[\sign]$.
	Lastly, we define the endpoints $\pathleft[\freesign{\sign}{A}]_j$ and $\pathright[\freesign{\sign}{A}]_j$ to be $\freepath{\pathleft[\sign]_j}$ and $\freepath{\pathright[\sign]_j}$ respectively.
	\item To define the homotopy endpoints of $\freesign{\sign}{A}$, we first note that each $h : \homotendpoint{\pathleft[\sign]}{\pathright[\sign]}{a}{b}{s}{t}$ gives rise to a homotopy endpoint $\freehomot{h} : \homotendpoint{\pathleft[\freesign{\sign}{A}]}{\pathright[\freesign{\sign}{A}]}{\freepath{a}}{\freepath{b}}{\freepath{s}}{\freepath{t}}$.
	Now we define $\homotleft[\freesign{\sign}{A}]$ and $\homotright[\freesign{\sign}{A}]$ to be $\freehomot{\homotleft[\sign]}$ and $\freehomot{\homotright[\sign]}$ respectively.
\end{itemize}
We define the \fat{free $\sign$-algebra} on $A$ to be the biinitial $\freesign{\sign}{A}$-algebra.
The free $\sign$-algebra on $A$ is denoted by $\freealg{\sign}{A}$ and the inclusion is denoted by $\freealginc{\sign} : A \rightarrow \freealg{\sign}{A}$.
\end{defi}

Note that the free algebra exists by Construction \ref{constr:hit_exist}.
The free signature of the signature for monoidal objects (Subsection \ref{sec:monoidal_object}) is given as follows

\begin{lstlisting}[mathescape=true]
Inductive $\freemonoid{A}$ :=
| $\freemonoidinc$ : $A \rightarrow \freemonoid{A}$
| $\monobjunit$ : $\freemonoid{A}$
| $\monobjtensor$ : $\freemonoid{A} \rightarrow \freemonoid{A} \rightarrow \freemonoid{A}$
| $\monobjlambda$ : $\depprod{(x : \freemonoid{A})}{\monobjtensor(\monobjunit,x) = x}$
| $\monobjrho$ : $\depprod{(x : \freemonoid{A})}{\monobjtensor(x,\monobjunit) = x}$
| $\monobjalpha$ : $\depprod{(x ,y,z: \freemonoid{A})}{\monobjtensor(x,(\monobjtensor(y,z)))=\monobjtensor(\monobjtensor(x,y),z)}$
| $\monobjtr$ : $\depprod{(x,y : \freemonoid{A})}{\ap{(\Lam z. \monobjtensor(x,z))}{(\monobjlambda(y))}} = \concat{\monobjalpha(x,\monobjunit,y)}{\ap{(\Lam z. \monobjtensor(z,y))}{(\monobjrho(x))}}$
| $\monobjpent$ : $\depprod{(w,x,y,z : \freemonoid{A})}{}$
    $\concat{\monobjalpha(w,x,\monobjtensor(y,z))}{\monobjalpha(\monobjtensor(w,x),y,z)}$
    =
    $\concat{\ap{(\Lam v. \monobjtensor(w,v))}{(\monobjalpha(x,y,z))}}{\concat{\monobjalpha(w,\monobjtensor(x,y),z)}{\ap{(\Lam v. \monobjtensor(v,z))}{(\monobjalpha(w,x,y))}}}$
\end{lstlisting}

Note that this corresponds to the free monoidal groupoid as given by Piceghello \cite{Piceghello19}.
To construct the desired left biadjoint pseudofunctor, we first give a more convenient formulation of the biinitiality of $\freealg{\sign}{A}$.
By stating the mapping properties with algebras for $\sign$, we can apply them directly when constructing the biadjunction.

\begin{cor}
\label{cor:free_alg_ump}
Let $\sign$ be a signature and let $A$ be a 1-type.
Suppose that we have $Y : \algM(\sign)$ and $\iota_Y : A \rightarrow Y$.
\begin{itemize}
	\item There is an algebra morphism $f : \freealg{\sign}{A} \onecell Y$ such that for all $a : A$, we have $f(\freealginc{\sign}(a)) = \iota_Y(a)$.
	\item
	Given algebra morphisms $f_1, f_2 : \freealg{\sign}{A} \onecell Y$
	and paths $p_1 : f_1(\freealginc{\sign}(a)) = \iota_Y(a)$ and $p_2 : f_2(\freealginc{\sign}(a)) = \iota_Y(a)$ for each $a : A$,
	there is a unique 2-cell $\tc : f_1 \twocell f_2$
	such that $\tc(\freealginc{\sign}(a)) \vcomp \AlgMapPoint{f}(\inr(a)) = \AlgMapPoint{g}(\inr(a))$ for each $a : A$.
\end{itemize} 
\end{cor}

Now we have enough in place to construct the desired biadjunction.

\begin{prob}
\label{prob:free_alg}
Given a signature $\sign$, to construct a coherent biadjunction $\freealgpsfun{\sign} \dashv \underlying_{\sign}$ where $\underlying_{\sign} : \pseudo(\algM(\sign), \onetypes)$ is the forgetful pseudofunctor.
\end{prob}

\begin{construction}{prob:free_alg}
We only indicate how to construct the left biadjoint.
To construct $\freealgpsfun{\sign} : \pseudo(\onetypes, \algM(\sign))$, we need to give an algebra $\freealgpsfun{\sign}(A)$ for $A : \onetypes$.
We define $\freealgpsfun{\sign}(A)$ to be the biinitial algebra for $\freesign{\sign}{A}$.
Pseudofunctoriality follows from Corollary \ref{cor:free_alg_ump}.
\end{construction}

From all of this, we conclude that each signature gives rise to a pseudomonad on 1-types.

\begin{prop}[Proposition 5.1 in \cite{LACK2000179}]
Each coherent biadjunction gives rise to a pseudomonad.
\end{prop}

\begin{cor}
\label{cor:pseudomonadofsig}
Each signature gives rise to a pseudomonad on 1-types.
\end{cor}

\section{The First Isomorphism Theorem}
\label{sec:isomorphism_theorem}
The first isomorphism theorem is one of the classical results in universal algebra.
While this statement is usually about algebraic structures in sets \cite{lynge2019}, we look at a generalization to the 1-truncated case.
More specifically, we formulate and prove this theorem for algebras for the signatures defined in Definition \ref{def:signature}.

Before stating and proving the isomorphism theorem, we need to generalize several notions from universal algebra to the 1-truncated case.
First of all, following the approach of Section \ref{sec:finite_limits}, we define the image of an algebra morphism using displayed algebras.
Second of all, we define congruence relations and their quotients.
The main idea here is to use Construction \ref{constr:alg_biadj}, where
we lifted the groupoid quotient to a biadjunction from $\algG{\sign}$ to $\algebra{\sign}$.
This result also indicates how congruence relations will be defined.
Basically, a congruence relation on $X : \algebra{\sign}$ is a groupoid structure on $X$ which gives an algebra of groupoids.
Once these notions are in place, we can formulate and prove the first isomorphism theorem similarly to the set-theoretical version.

\subsection{The Image}
We start by defining the image of an algebra morphism.
For each $f : A \onecell B$ we are after a factorization $A \onecell \Im{f} \onecell B$.

\begin{prob}
\label{prob:image}
Given algebras $A$ and $B$ for a signature $\sign$ and an algebra morphism $f : A \onecell B$,
to construct an algebra $\Im(f)$ and morphisms $\ImProj{f} : A \onecell \Im{f}$ and $\ImInc{f} : \Im(f) \onecell B$. 
\end{prob}

\begin{construction}{prob:image}
We define $\Im(f)$ to be the total algebra $\TotalAlg Y$ of the following displayed algebra $Y$ over $B$:
\begin{itemize}
\item The underlying family of 1-types over $B$ is $Y(x) \eqdef \exists a : A.\ f(a) = x$.
\item For each $x : \polyAct{\pointconstr}{A}$, we are required to construct
a map $\DispAlgPoint{Y} : \polyDact{\pointconstr}{Y}(x) \rightarrow
\exists a : A.\ f(a) = \AlgPoint{B}(x)$. Suppose we have $\pover{x}
: \polyDact{\pointconstr}{Y}(x)$. By induction on the polynomial
$\pointconstr$, it is possible to derive from $\pover{x}$ an
inhabitant of $\exists y
: \polyAct{\pointconstr}{A}.\ \polyAct{\pointconstr}{f}(y) = x$. By
invoking the elimination principle of propositional truncation, it is sufficient to
construct a map $f : \Sum {y
: \polyAct{\pointconstr}{A}}.\ \polyAct{\pointconstr}{f}(y) =
x \rightarrow \Sum {a : A}.\ f(a) = \AlgPoint{B}(x)$ in
order to define the map $\DispAlgPoint{Y}$.  So, assume that we have
$y : \polyAct{\pointconstr}{A}$ and a path $p
: \polyAct{\pointconstr}{f}(y) = x$. Take $f(y,p)$ to be the pair
consisting of $\AlgPoint{A}(y)$ and the path
\[
f(\AlgPoint{A}(y))
\stackrel{\AlgMapPoint{f}(y)}{=}
\AlgPoint{B}(\polyAct{\pointconstr}{f}(y))
\stackrel{\ap {\AlgPoint{B}} p}{=}
\AlgPoint{B}(x)
\]
\item Since $Y(x)$ is a proposition for all $x : B$, the construction of
paths $\DispAlgPath{Y}{j}$ and globes $\DispAlgHomot{Y}{j}$ is
straightforward.
\end{itemize}

The map $\ImInc{f}$ is the first projection and $\ImProj{f}$ sends $x : A$ to
$(f(x),\PC(x,\refl{f(x)}))$, where $\PC$ is the point constructor of
propositional truncation. 
\end{construction}

\subsection{Congruence Relations}
Next we define congruence relations and we show that each such relation gives an algebra in groupoids.
The difficulty here is constructing the homotopy constructor of that algebra.
For this reason, we define congruence relations in two steps, and we start by defining \emph{path congruence relations} and showing that these give rise to path algebras in groupoids.

\begin{defi}
Let $\sign$ be a signature and let $X : \algebra{\sign}$.
A \fat{path congruence relation} $R$ on $X$ consists of
\begin{itemize}
	\item a groupoid structure $R$ on $X$;
	\item for each $x, y : \polyAct{\pointconstr[\sign]}{X}$ and $f : \mor{\polyAct{\pointconstr[\sign]}{R}}{x}{y}$, a morphism $\AlgPoint{R} : \mor{R}{\AlgPoint{X}(x)}{\AlgPoint{X}(y)}$ and a proof that this assignment is functorial;
	\item a proof that the assignment $\Lam x, \idtoiso(\AlgPath{X}{j}(x))$ gives a natural transformation $\semEG{\pathleft[\sign]_j} \twocell \semEG{\pathright[\sign]_j}$ for each $j : \pathlabel[\sign]$.
\end{itemize}
\end{defi}

\begin{prob}
\label{prob:path_congruence_relation_to_grpd}
Given a path congruence relation $R$,
to construct $\toGrpdPathAlg{R} : \pathalgG{\sign}$.
\end{prob}

\begin{construction}{prob:path_congruence_relation_to_grpd}
We only discuss the data involved.
\begin{itemize}
	\item The carrier $\toGrpdPathAlg{R}$ is the groupoid whose type of objects is $X$ and whose morphisms are given by $R$.
	\item Next we define a functor $\semPG{\pointconstr[\sign]}(\toGrpdPathAlg{R}) \onecell \toGrpdPathAlg{R}$. It is defined to be $\AlgPoint{X}$ on objects and $\AlgPoint{R}$ on morphisms.
	\item Lastly, we define a natural transformation $\semEG{\pathleft[\sign]_j} \twocell \semEG{\pathright[\sign]_j}$ for each $j : \pathlabel[\sign]$. The component function of this transformation is $\Lam x, \idtoiso(\AlgPath{X}{j}(x))$. \qedhere
\end{itemize}
\end{construction}

To show that $\toGrpdPathAlg{R}$ is an algebra, we also need to give the homotopy constructor.
Since this constructor has both a point and a path constructor, we need to check an equality of morphisms of $\toGrpdPathAlg{R}$, which depend on points of $X$ and morphisms of $\toGrpdPathAlg{R}$.
As a result, we cannot reuse the homotopy constructor of $X$, because it only depends on points of $X$ and paths in $X$.
This means that to construct the groupoid algebra, the homotopy constructor needs to be constructed from scratch.
This is reflected in the following definition.

\begin{defi}
A path congruence relation $R$ is a \fat{congruence} relation if $\toGrpdPathAlg{R}$ is an algebra in groupoids.
\end{defi}

Note that each congruence relation $R$ gives rise to an algebra $\toGrpdAlg{R} : \algG{\sign}$.
By Construction \ref{constr:alg_biadj}, the groupoid quotient lifts to $\alggquot : \pseudo(\algG{\sign}, \algebra{\sign})$.
Hence, we can construct an algebra $\alggquot(\toGrpdAlg{R}) : \algebra{\sign}$ from a congruence relation $R$.
Note that the map $\gcl : X \rightarrow \alggquot(\toGrpdAlg{R})$ is a homomorphism of algebras.
Since $\alggquot$ is a left biadjoint, we also get a mapping property for $\alggquot(\toGrpdAlg{R})$.

\begin{rem}
\label{remark:mapping_property_congruence}
Suppose that we have two algebras $X, Y : \algebra{\sign}$ and a congruence relation $R$ on $X$.
From the biadjunction, we get a unique morphism $\overline{g} : \alggquot{\toGrpdAlg{R}} \onecell Y$ from a functor $g : R \onecell \algpgrpd{Y}$ such that $\overline{g}(\gcl(x)) = g(x)$ and $\ap{\overline{g}}{(\gcleq(p))} = g(p)$ where $p$ is a morphism in $R$ from $x$ to $y$.

Using this, we can give conditions for when an algebra morphism $f : X \onecell Y$ factors through $\alggquot{\toGrpdAlg{R}}$.
We can unfold the definition to find the necessary ingredients for the factorization.
For example, we need to show that $f$ lifts to a functor, which means we need to provide a path $f_2 : f(x) = f(y)$ for each $x, y : X$ and $p : \mor{R}{x}{y}$, and equalities $f_2(\idpath(x)) =\idpath(f_2(x))$ and $f_2(p \cdot q) = f_2(p) \vcomp f_2(q)$.
\end{rem}

\subsection{The First Isomorphism Theorem}
Now let us prove a generalization of the first isomorphism theorem to the 1-truncated case.
Note that our proof follows the same steps as the proof of the first isomorphism theorem for sets.
Let us start by characterizing adjoint equivalences in $\algM(\sign)$.

\begin{prop}
\label{prop:algebra_adjequiv}
Suppose we have a signature $\sign$, algebras $X, Y : \algM(\sign)$, and an algebra morphism $f : X \onecell Y$.
Then $f$ is an adjoint equivalence in $\algM(\sign)$ if its carrier is an adjoint equivalence in $\onetypes$.
\end{prop}

Before we state and prove the isomorphism theorem, we define the kernel of algebra morphisms.

\begin{defi}[Kernel]
Given a signature $\sign$, algebras $X, Y : \algM(\sign)$, and an algebra morphism $f : X \onecell Y$, we define a groupoid algebra $\Ker(f)$ on $X$, called the \fat{kernel} of $f$, by setting $R(x, y) \eqdef f(x) = f(y)$.
\end{defi}

\begin{thm}\label{thm:iso_thm}
Let $\sign$ be a signature and suppose we have algebras $X, Y : \algM(\sign)$ and an algebra morphism $f : X \onecell Y$.
Then
\begin{itemize}
	\item there is a unique map $\overline{f} : \alggquot(\toGrpdAlg{\Ker(f)}) \onecell Y$ such that we have propositional equalities $p : \depprod{(x : X)}{\overline{f}(\gcl(x)) = f(x)}$ and $\concat{\ap{\overline{f}}{(\gcleq(r))}}{p(y)} = \concat{p(x)}{r}$ where $r : f(x) = f(y)$;
	\item there is an adjoint equivalence $\widetilde{f} : \alggquot(\toGrpdAlg{\Ker(f)}) \onecell \Im(f)$ such that $\pi_1(\widetilde{f}(x)) = \overline{f}(x)$.
\end{itemize}
\end{thm}

\begin{proof}
Using Remark \ref{remark:mapping_property_congruence}, we can define the maps $\overline{f} : \alggquot(\toGrpdAlg{\Ker(f)}) \onecell Y$ and $\widetilde{f} : \alggquot(\toGrpdAlg{\Ker(f)}) \onecell \Im(f)$.
Note that from Remark \ref{remark:mapping_property_congruence}, we also get that $\overline{f}$ satisfies the required equalities and that $\overline{f}$ is unique.
To show that $\overline{f}$ is an adjoint equivalence, we use Proposition \ref{prop:algebra_adjequiv}, and we do that by showing that the fibers are contractible.
Proving that the fibers are inhabited and propositional is similar to proving the surjectivity and injectivity for the set-theoretical first isomorphism theorem \cite{lynge2019}.
\end{proof}

\section{Calculating Fundamental Groups}
\label{sec:fundamental_groups}
Let us finish by using Construction \ref{constr:hit_exist} to determine the fundamental group of some HITs.
Such results are often proven by encode-decode method \cite{LicataS13,LicataF14}, but we take a different approach.
We only use encode-decode to determine the path space of the groupoid quotient.
For the HITs considered in this section, we give a simpler description of the initial groupoid algebra, which fixes the fundamental group.
We start by determining the path space of the groupoid quotient.

\begin{prop}
\label{prop:groupoid_quot_encode_decode}
Let $G$ be a groupoid and let $x$ and $y$ be objects in $G$.
Then the types $\gcl(x) = \gcl(y)$ and $\mor{G}{x}{y}$ are equivalent.
\end{prop}

As a result, the type of paths between two points in a HIT is the type of morphisms in the initial groupoid algebra.
Since this algebra is unique, we can determine the fundamental group by finding a simpler description of the initial groupoid algebra.

\subsection{Circle}
\label{sec:circle_fund_group}
Recall that the circle is defined as the following HIT

\begin{lstlisting}[mathescape=true]
Inductive $\circleS$ :=
| $\baseS$ : $\circleS$
| $\SLoop$ : $\baseS = \baseS$
\end{lstlisting}

We can define a signature $\circleS$ that represents this HIT. Notice that the HIT generated by this signature includes also
a 1-truncation constructor, which is superflous in this case since $\circleS$ is already provably
1-truncated.

Next we construct a groupoid algebra $\circlegrpd$ of this signature and we prove that $\circlegrpd$ is biinitial.

\begin{defi}
We define a groupoid $\circlegrpd$ as follows
\begin{itemize}
	\item the type of objects is the unit type;
	\item the type of morphisms from $\unitt$ to $\unitt$ is the type integers.
\end{itemize}
\end{defi}

\begin{prob}
\label{prob:initial_grpd_alg_circle}
To construct an $\circleS$-algebra structure on $\circlegrpd$.
\end{prob}

\begin{construction}{prob:initial_grpd_alg_circle}
To construct the desired algebra structure, we first need to define a functor $\circlegrpdbase$ from the unit category to $\circlegrpd$.
This functor sends the unique element to $\unitt$.
Furthermore, we need to define a natural transformation $\circlegrpdloop$ from $\circlegrpdbase$ to $\circlegrpdbase$.
On each component, this transformation is defined to be $1$.
\end{construction}

\begin{prop}
The groupoid $\circlegrpd$ is biinitial in $\algG{\circleS}$.
\end{prop}

From Construction \ref{constr:hit_exist} and the fact that biinitial objects are unique up to equivalence,
we can deduce that the circle is the groupoid quotient of $\circlegrpd$ and that its base point $\baseS$ is $\gcl(\unitt)$.
Since the morphisms of $\circlegrpd$ are just the integers, we immediately get the following from Proposition \ref{prop:groupoid_quot_encode_decode}

\begin{cor}
The type $\baseS = \baseS$ is equivalent to the integers.
\end{cor}

\subsection{Torus}
\label{sec:torus_fund_group}
Next we look at the torus, which we defined in Example \ref{ex:torus}.
We use the same approach to determine its fundamental group, so we start by giving a simpler description of the biinitial algebra for $\torus$ in groupoids.

\begin{defi}
Define a groupoid $\torusgrpd$ as follows
\begin{itemize}
	\item the type of objects is the unit type;
	\item the type of morphisms from $\unitt$ to $\unitt$ is $\mathbb{Z} \times \mathbb{Z}$.
\end{itemize}
\end{defi}

\begin{prob}
\label{prob:initial_grpd_alg_torus}
To construct an $\torus$-algebra structure on $\torusgrpd$.
\end{prob}

\begin{construction}{prob:initial_grpd_alg_torus}
\label{constr:initial_grpd_alg_torus}
Again the point constructor is given by a functor $\torusgrpdbase$ from the unit category to $\torusgrpd$, which sends the unique element to $\unitt$.
The two natural transformations $\torusgrpdloopl$ and $\torusgrpdloopr$ from $\torusgrpdbase$ to $\torusgrpdbase$ send each element to $(1, 0)$ and $(0, 1)$ respectively.
For the last component, we need to give an equality $\torusgrpdsurf$ between $(1, 0) + (0, 1) = (0, 1) + (1, 0)$.
This holds definitionally.
\end{construction}

\begin{prop}
The groupoid $\torusgrpd$ is biinitial in $\algG{\torus}$.
\end{prop}

Again we can deduce that the torus is the groupoid quotient of $\torusgrpd$ and that its base point $\base$ is $\gcl(\unitt)$.
From Proposition \ref{prop:groupoid_quot_encode_decode}, we immediately get

\begin{cor}
The type $\base = \base$ is equivalent to $\mathbb{Z} \times \mathbb{Z}$.
\end{cor}

\subsection{Group Quotient}
\label{sec:group_quotient_fund_group}
For the remainder of this section, we assume that a group $G$ is given.
Our goal is to show that the fundamental group of the group quotient at its base point is $G$.
Again we do this by determining the biinitial groupoid algebra.
The algebra structure is constructed in a similar fashion to Construction \ref{constr:initial_grpd_alg_torus}.

\begin{defi}
We define a groupoid $\grquotgrpd$ as follows
\begin{itemize}
	\item the type of objects is the unit type;
	\item the type of morphisms from $\unitt$ to $\unitt$ is $G$.
\end{itemize}
\end{defi}

\begin{prob}
\label{prob:initial_grpd_alg_group_quot}
To construct an $\groupquot{G}$-algebra structure on $\grquotgrpd$.
\end{prob}

\begin{prop}
The groupoid $\grquotgrpd$ is biinitial in $\algG{\groupquot{G}}$.
\end{prop}

\begin{cor}
The fundamental group of the group quotient of $G$ is $G$ itself.
\end{cor}

\section{Conclusion and Further Work}
\label{sec:conclusion}
We showed how to construct finitary 1-truncated higher inductive types using the propositional truncation, quotient, and the groupoid quotient.
This reduces the existence of a general class of HITs to simpler ones.
We needed the types to be 1-truncated, so that we could use the framework of bicategory theory,
and the HITs we studied had to be finitary to guarantee that the groupoid quotient
commutes with the involved operations \cite{DBLP:journals/mscs/ChapmanUV19}.
On the way, we also proved that HITs are unique and we studied universal algebra with our signatures.
We showed that the bicategory of algebras has finite limits and we proved the first isomorphism theorem for these algebras.
Lastly, we used the way we constructed HITs to calculate fundamental groups.

There are numerous ways to improve on these results.
First of all, we only constructed finite limits of algebras while it should also be possible to construct finite colimits of algebras.
The scheme studied in this paper is not flexible enough to support these colimits since we do not have a path endpoint that represents the action of a polynomial on the point constructor.
Hence, if we want to internally construct these colimits, then we need to define a more permissive signature for higher inductive types.

Secondly, it should be possible to modify our approach to obtain HITs in directed type theory (DTT) \cite{north2019towards}.
In the model of DTT provided by North, types are interpreted as categories and higher inductive types in DTT could be interpreted as initial algebras.
We constructed such algebras in the bicategory of groupoids and in a similar way, one should be able to construct the desired algebras in the bicategory of categories.

Lastly, our construction only considers a rather simple scheme of HITs.
In particular, we restrict ourselves to the 1-truncated case. 
Since untruncated types correspond to $\infty$-groupoids,
generalizing the methods used in this paper to the untruncated case,
requires formalizing notions from $\infty$-category theory in type theory
\cite{DBLP:conf/csl/AltenkirchR12,DBLP:journals/pacmpl/CapriottiK18,FinsterM17}.
This also requires finding an $\infty$-dimensional generalization of the groupoid quotient.
An alternative approach to deal with untruncated HITs, pointed out by Ali Caglayan,
would be using wild categories \cite{DBLP:conf/tlca/HirschowitzHT15,KrausRaumer}.
Note that generalizing Construction \ref{constr:initial_grpd_alg} would also pose a challenge, because the set quotient cannot be used.
Instead all $n$-morphisms must be freely generated by the $n$-path constructors.
We would also like to extend our scheme to incorporate both indexed HITs and higher inductive-inductive types \cite{CavalloH19,KaposiK18}.

\section*{Acknowledgments}
The authors thank Herman Geuvers, Dan Frumin, Benedikt Ahrens, and Ali Caglayan for helpful comments and discussions.
The authors also thank the anonymous reviewers for their helpful comments and suggestions.
Niccol{\`o} Veltri was supported by the Estonian Research Council grant PSG659 and by the ESF funded
Estonian IT Academy research measure (project 2014-2020.4.05.19-0001).

\bibliographystyle{alpha}
\bibliography{literature}

\end{document}